%% file: main.tex
\documentclass[10pt]{article}

\input{macros}

\newcommand{\dom}{\problemname{DOM}} 
\newcommand{\bcqc}{\problemname{BagCQC}} 
\newcommand{\bcqca}{\problemname{BagCQC-A}} 
\newcommand{\ii}{\problemname{II}} 
\newcommand{\mii}{\problemname{Max-II}} 
\newcommand{\iip}{\problemname{IIP}} 
\newcommand{\miip}{\problemname{Max-IIP}} 
\newcommand{\umiip}{\problemname{Uniform-Max-IIP}} 

\newcommand{\calM}{\mathcal M} 
\newcommand{\calN}{\mathcal N} 

\allowdisplaybreaks

\begin{document}

\title{Bag Query Containment  and Information Theory}

\author{
    Mahmoud Abo Khamis \\ {\small relational\underline{AI}} \and
    Phokion G.\ Kolaitis \\ {\small UC Santa Cruz} \and
    Hung Q. Ngo \\ {\small relational\underline{AI}} \and
    Dan Suciu \\ {\small University of Washington}
}

\date{}

\maketitle

\input{abstract}

\input{sec-intro}
\input{sec-problem-definition}
\input{sec-main-results}
\input{sec-sufficient-and-necessary}
\input{sec-completeness}
\input{sec-decidability}
\input{sec-discussion}

\section*{Acknowledgments}
Suciu was partially supported by NSF grants III-1703281, III-1614738,
IIS-1907997, AitF-1535565.
Kolaitis was partially supported by NSF Grant IIS-1814152.

\bibliographystyle{acm}
\bibliography{main}

\appendix

\input{appendix}

\end{document}

%% file: macros.tex
\usepackage{latexsym}
\usepackage{amsmath}
\usepackage{amssymb}
\usepackage{amsfonts}
\usepackage{amsthm}
\usepackage{cite}
\usepackage{graphicx}
\usepackage{tikz,tikz-qtree,tikz-qtree-compat}
\usepackage{cite,color,xspace,colortbl}
\usepackage{algorithm}
\usepackage[noend]{algpseudocode}
\usepackage{fullpage}

\newcommand{\conferenceorfull}[2]{#2}

\usepackage{array}
\newcolumntype{L}[1]{>{\raggedright\let\newline\\\arraybackslash\hspace{0pt}}p{#1}}
\newcolumntype{C}[1]{>{\centering\let\newline\\\arraybackslash\hspace{0pt}}p{#1}}
\newcolumntype{R}[1]{>{\raggedleft\let\newline\\\arraybackslash\hspace{0pt}}p{#1}}

\usepackage{tikz,tikz-qtree,tikz-qtree-compat,tikz-3dplot}

\usetikzlibrary{shapes,decorations.markings,arrows.meta,fit,calc,positioning}

\usepackage{todonotes}
\usepackage{url}
\usepackage{bm}

\colorlet{LightViolet}{violet!40}
\colorlet{LightRed}{red!40}
\colorlet{LightOrange}{orange!40}
\colorlet{LightGreen}{green!40}
\colorlet{LightBlue}{blue!40}
\colorlet{DarkGreen}{green!50!black}
\colorlet{DarkRed}{red!70!black}
\colorlet{DarkCyan}{red!70!black}
\colorlet{DarkBlue}{blue!80!black}
{\definecolor{DarkOrange}{rgb}{1.0, 0.49, 0.0}
\definecolor{Airforceblue}{rgb}{0.36, 0.54, 0.66}

\newcommand{\defeq}{\stackrel{\mathrm{def}}{=}}

\renewcommand{\hom}{\text{\sf hom}}
\newcommand{\td}{\text{\sf TD}}

\newcommand{\parent}{\text{\sf parent}}

\newcommand{\functionname}[1]{\text{\sf #1}}

\newcommand{\atoms}{\functionname{atoms}}

\newcommand{\vars}{\functionname{vars}}
\newcommand{\rel}{\functionname{rel}}

\newcommand{\arity}{\functionname{arity}}

\newcommand{\problemname}[1]{\text{\tt #1} }

\newcommand{\calA}{\mathcal A}
\newcommand{\calB}{\mathcal B}

\newcommand{\calD}{\mathcal D}

\newcommand{\calR}{\mathcal R}

\newcommand{\calI}{\mathcal I}

\newcommand{\Z}{\mathbb Z} 
\newcommand{\Q}{\mathbb Q} 
\newcommand{\R}{\mathbb R} 

\newcommand{\pr}{\mathop{\textnormal{Pr}}}

\newcommand{\be}{\begin{enumerate}}
\newcommand{\ee}{\end{enumerate}}
\newcommand{\bi}{\begin{itemize}}
\newcommand{\ei}{\end{itemize}}
\newcommand{\beq}{\begin{equation}}
\newcommand{\eeq}{\end{equation}}

\newcommand{\bp}{\begin{proof}}
\newcommand{\ep}{\end{proof}}
\newcommand{\bcor}{\begin{cor}}
\newcommand{\ecor}{\end{cor}}
\newcommand{\bthm}{\begin{thm}}
\newcommand{\ethm}{\end{thm}}
\newcommand{\blmm}{\begin{lmm}}
\newcommand{\elmm}{\end{lmm}}
\newcommand{\bdefn}{\begin{defn}}
\newcommand{\edefn}{\end{defn}}
\newcommand{\bprop}{\begin{prop}}
\newcommand{\eprop}{\end{prop}}
\newcommand{\bconj}{\begin{conj}}
\newcommand{\econj}{\end{conj}}
\newcommand{\bopm}{\begin{opm}}
\newcommand{\eopm}{\end{opm}}
\newcommand{\bprm}{\begin{prm}}
\newcommand{\eprm}{\end{prm}}
\newcommand{\brmk}{\begin{rmk}}
\newcommand{\ermk}{\end{rmk}}
\newcommand{\bclm}{\begin{claim}}
\newcommand{\eclm}{\end{claim}}
\newcommand{\bex}{\begin{ex}}
\newcommand{\eex}{\end{ex}}

\newcommand{\inner}[1]{\langle #1 \rangle}

\newcommand{\mv}[1]{\mathbf{#1}}

\theoremstyle{plain}                   
\newtheorem{thm}{Theorem}[section]
\newtheorem{lmm}[thm]{Lemma}
\newtheorem{prop}[thm]{Proposition}
\newtheorem{cor}[thm]{Corollary}

\theoremstyle{definition}              
\newtheorem{pbm}[thm]{Problem}

\newtheorem{opm}[thm]{Open Problem}
\newtheorem{prm}[thm]{Problem}
\newtheorem{conj}[thm]{Conjecture}
\newtheorem{ex}[thm]{Example}

\newtheorem{defn}[thm]{Definition}

\newtheorem{rmk}[thm]{Remark}
\newtheorem{fact}[thm]{Fact}
\newcounter{claimCounter}
\newtheorem{claim}[claimCounter]{Claim}

\newtheorem*{example*}{Example}

\newtheorem*{lmm*}{Lemma}

\definecolor{Red}{RGB}{255,204,204}
\definecolor{Green}{RGB}{204,255,204}
\definecolor{Blue}{RGB}{204,204,255}

\newcommand{\set}[1]{\{#1\}}                    
\newcommand{\setof}[2]{\{{#1}\mid{#2}\}}        

\newcommand{\nodes}{\texttt{nodes}}
\newcommand{\edges}{\texttt{edges}}
\newcommand{\eat}[1]{}

%% file: abstract.tex
\begin{abstract}
The query containment problem is a fundamental algorithmic problem in data management. While
this problem is well understood under set semantics, it is by far less understood under bag
semantics. In particular, it is a long-standing open question whether or not  the
conjunctive query containment problem  under bag semantics is decidable. We unveil tight
connections between information theory and the conjunctive query containment under bag
semantics. These connections are established using information inequalities, which are
considered to be the laws of information theory. Our first main result asserts that deciding
the validity of a generalization of information inequalities is many-one equivalent to the restricted case of conjunctive query containment in which the containing query is acyclic; thus, either both these problems are decidable or both are undecidable.  Our second main result identifies a new decidable case of the conjunctive query containment problem under bag semantics. Specifically, we give an exponential-time algorithm for conjunctive query containment under bag semantics, provided the containing query is chordal and admits a simple junction tree.
\end{abstract}

%% file: sec-intro.tex
\section{Introduction}

\label{sec:intro}

Since the early days of relational databases, the query containment problem has been recognized as a fundamental algorithmic problem  in data management.
This problem asks:
given two queries $Q_1$ and $Q_2$, is it true that $Q_1(\calD)\subseteq
Q_2(\calD)$, for every database $\calD$? Here, $Q_i(\calD)$ is the result of evaluating the query $Q_i$ on the database $\calD$. Thus, the query containment problem has several different variants,  depending on whether the evaluation uses set semantics or bag semantics, and  whether $\calD$ is a set database or a bag database.
Query containment  under set semantics on set databases is the most extensively studied and well understood such variant. In particular, Chandra and Merlin \cite{DBLP:conf/stoc/ChandraM77} showed that, for this variant, the containment problem for conjunctive queries is NP-complete.

Chaudhuri and Vardi \cite{DBLP:conf/pods/ChaudhuriV93} were the first to raise the importance of studying the query containment problem under bag semantics. In particular, they raised the question of the decidability of the containment problem  for conjunctive queries under bag semantics. There are two variants of this problem: in the \emph{bag-bag} variant,  the evaluation uses bag semantics and the input database is a bag, while in the \emph{bag-set} variant, the evaluation uses bag semantics and the input database is a set. It is known that for conjunctive queries, the bag-bag variant and the bag-set variant are polynomial-time reducible to each other (see, e.g., \cite{DBLP:conf/pods/JayramKV06}); in particular, either both variants are decidable or both are undecidable. Which of the two is the case, however, remains an outstanding open question to date.

During the past twenty five years, the research on the query containment problem under bag semantics has produced a number of results about extensions of conjunctive queries and also about restricted classes of conjunctive queries. Specifically, using different reductions from Hilbert's 10th Problem, it has been shown that the containment problem under bag semantics is undecidable for both the class of unions of conjunctive queries \cite{DBLP:journals/tods/IoannidisR95} and  the class of conjunctive queries with inequalities~\cite{DBLP:conf/pods/JayramKV06}. It should be noted that, under set semantics, the containment problem for these two classes of queries is decidable; in fact, it is NP-complete for unions of conjunctive queries \cite{DBLP:journals/jacm/SagivY80}, and  it is $\Pi_2^{\text{P}}$-complete for conjunctive queries with inequalities \cite{DBLP:journals/jacm/Klug88,DBLP:journals/jcss/Meyden97}.
As regards to restricted classes of conjunctive queries, several decidable cases of the bag-bag variant were identified in
\cite{DBLP:journals/ipl/AfratiDG10}, including the case where both $Q_1$ and $Q_2$ are \emph{projection-free} conjunctive queries, i.e., no variable is existentially quantified. Quite recently, this decidability result was extended to the case where $Q_1$ is a projection-free conjunctive query and $Q_2$ is an arbitrary conjunctive query \cite{KM2019}; the proof is via a reduction to a decidable class of Diophantine inequalities. In a different direction, information-theoretic methods were used in \cite{HDE} to study the \emph{homomorphism domination exponent} problem, which generalizes the conjunctive query containment problem under bag semantics on graphs. In particular, it was shown in \cite{HDE} that the conjunctive query containment problem under bag semantics is decidable when  $Q_1$ is a series-parallel graph and $Q_2$ is a chordal graph. This was the first time that notions and techniques from information theory were applied to the study of the containment problem under bag semantics.

Notions and techniques from information theory have  found a number of uses in other
areas of database theory. For example, entropy and mutual information have been used to
characterize database dependencies \cite{DBLP:journals/tse/Lee87,DBLP:journals/tse/Lee87a}
and normal forms in relational and XML databases \cite{DBLP:journals/jacm/ArenasL05}. More
recently, information inequalities were used with much success to obtain tight bounds
on the size of the output of a query on a given database
\cite{DBLP:journals/siamcomp/AtseriasGM13,DBLP:journals/jacm/GottlobLVV12,
DBLP:journals/talg/GroheM14,DBLP:conf/pods/KhamisNS16,DBLP:conf/pods/Khamis0S17},
{\em and} to devise query plans for worst-case optimal join
algorithms~\cite{DBLP:conf/pods/KhamisNS16,DBLP:conf/pods/Khamis0S17}.

This paper unveils deeper connections between information theory and the query containment problem under bag semantics. These connections are established through the systematic use of information inequalities, which have been called  the ``laws of information theory'' \cite{pippenger1986} as they express constraints on the entropy and thus ``govern the impossibilities in information theory'' \cite{yeung2012book}.

An \emph{information inequality} is an inequality of the form
\begin{align}
    0 &\leq  \sum_{X \subseteq V} c_X h(X), \label{eqn:ii}
\end{align}
where $V$ is a set of $n$ random variables over finite domains, each coefficient $c_X$ is a
real number, i.e.  $c = (c_X)_{X \subseteq  V}$ is a $2^n$-dimensional real vector,
$h$ is the {\em entropy function} of a joint distribution over $V$ ({\em $V$-distribution}
henceforth). In particular, $h(X)$ denotes the marginal entropy of the variables in the set
$X\subseteq V$.

An information inequality may hold for the entropy function of some $V$-distribution,
but may not hold for all $V$-distributions.  Following
\cite{DBLP:journals/entropy/Chan11}, we say that an information inequality is \emph{valid}
if it holds for the entropy function of {\em every} $V$-distribution. This notion gives rise
to the following natural decision problem, which we denote as $\iip$: given {\em integer}
coefficients $c_X \in \Z$ for all $X\subseteq V$, is the information inequality~\eqref{eqn:ii}
valid?\footnote{Equivalently, one can allow the input coefficients to be rational numbers.}

In this paper, we will also study a generalization of this problem that involves taking maxima of linear combinations of entropies. A  {\em max-information  inequality} is an expression of the form
\begin{align}
    0 &\leq  \max_{\ell \in [k]} \sum_{X \subseteq V} c_{\ell,X} h(X), \label{eqn:mii}
\end{align}
where  $V$,  $X$, $h(X)$ are as before, and for each $\ell \in [k]$,
$c_\ell := (c_{\ell, X})_{X \subseteq V}$ is a $2^n$-dimensional real vector.
We say that a max information inequality is \emph{valid} if it holds for the
entropy function of every $V$-distribution. We write $\miip$ to denote the following
decision problem: given $k$ integer vectors $c_\ell$ of dimension $2^n$,
is the  max information inequality \eqref{eqn:mii} valid?
Clearly, $\iip$ is the special case of $\miip$ in which $k=1$.

Our first main result asserts that $\miip$ is \emph{many-one equivalent} to the restricted
case of the conjunctive query containment problem  under bag semantics in which $Q_1$ is an
arbitrary conjunctive query and $Q_2$ is an acyclic conjunctive query. In fact, we show that
these two problems are reducible to each other via exponential-time many-one reductions.
This result establishes a {\em new} and {\em tight} connection between information theory and database theory, showing that  $\miip$ and the conjunctive query containment problem under bag semantics with acyclic $Q_2$ are equally hard.

To the best of our knowledge, it is not known whether $\miip$ is decidable. In fact, even $\iip$ is not known to be decidable; in other words, it is not known if there is an algorithm for telling whether a given  information inequality with integer coefficients is valid.  Even though the decidability question about $\iip$ and about $\miip$ does not seem to have been raised explicitly by researchers in information theory, we note that there is a growing body of research aiming to ``characterize'' all valid information inequalities; moreover,  finding such a ``characterization'' is regarded as a central problem in modern information theory (see, e.g., the survey \cite{DBLP:journals/entropy/Chan11}).  It is reasonable to expect that a ``good characterization'' of valid information inequalities will also give an algorithmic criterion for the validity of information inequalities. Thus, showing that $\iip$ is undecidable would imply that no ``good characterization'' of valid information inequalities exists.

Our second main result identifies a new decidable case of the conjunctive query containment
problem under bag semantics. Specifically, we show that there is an exponential-time
algorithm for testing whether $Q_1$ is contained in $Q_2$ under bag semantics, where $Q_1$
is an arbitrary conjunctive query and $Q_2$ is a conjunctive query that is \emph{chordal}
and admits a {\em junction tree} that is {\em simple}. Here, a query is chordal if its
Gaifman graph $G$ is chordal, i.e., $G$ admits a tree decomposition whose bags induce
(maximal) cliques of $G$; such a tree
   decomposition is called a {\em junction tree}. A tree decomposition is {\em simple} if
every pair of adjacent bags in the tree decomposition share at most one common variable.
The result follows from a new class of decidable $\miip$ problems.  Note that this result is
incomparable to the aforementioned decidability result about series-parallel and chordal
graphs in \cite{HDE}, in two ways. First, the result in \cite{HDE} applies only to graphs
(i.e., databases with a single binary relation symbol), while our result applies to
arbitrary relational schemas. Second, our result imposes more restrictions on $Q_2$, but no restrictions on $Q_1$.

The work reported here reveals that the conjunctive query containment problem under bag semantics is tightly intertwined with the validity problem for information inequalities. Thus, our work sheds new light on both these problems and, in particular, implies that
any  progress made in one of these problems will translate to similar progress in the other.


%% file: sec-problem-definition.tex
\section{Definitions}

\label{sec:problem}

We describe here the two problems whose connection forms the main
result of this paper.

\subsection{Query Containment Under Bag Semantics}
\label{sec:bcqc}

\paragraph{Homomorphisms between relational structures}
We fix a {\em relational vocabulary}, which is a tuple
$\calR = (R_1, \ldots, R_m)$, where each symbol $R_i$ has an
associated arity $a_i$. A {\em relational structure} is
$\calA = (A, R_1^A, \ldots, R_m^A)$, where $A$ is a finite set (called
domain) and each $R_i^A$ is a relation of arity $a_i$ over the
domain $A$.  Given two relational structures $\calA$ and $\calB$ with
domains $A$ and $B$ respectively, a homomorphism from $\calB$ to
$\calA$ is a function $f : B \rightarrow A$ such that for all $i$, we
have $f(R^B_i) \subseteq R_i^A$.  We write $\hom(\calB,\calA)$ for the
set of all homomorphisms from $\calB$ to $\calA$, and denote by
$|\hom(\calB,\calA)|$ its cardinality.

\paragraph{Bag-Set Semantics}
A {\em conjunctive query} $Q$ with variables $\vars(Q)$ and atom set
$\atoms(Q) = \set{A_1, \ldots, A_k}$ is a conjunction:
\begin{align}
    Q(\mv x) & = A_1 \wedge A_2 \wedge \cdots \wedge A_k. \label{eq:query:def}
\end{align}
For each $j \in [k]$, the atom $A_j$ is of the form
$R_{i_j}(\mv x_j)$, where $\rel(A_j) \defeq R_{i_j}$ is a relation
name, and $\vars(A_j) \defeq \mv x_j$ is a function,
\begin{align}
    \vars(A_j) &: [\arity(\rel(A_j))] \rightarrow \vars(Q) \label{eqn:vars:as:functions}
\end{align}
associating a variable to each attribute position of $\rel(A_j)$.  We
allow repeated variables in an atom.  The variables $\mv x$ are called
{\em head variables}, and must occur in the body.

A {\em database instance} is a structure $\calD$ with domain $D$. The
answer of a query \eqref{eq:query:def} with head variables $\mv x$ is
a set of $\mv x$-tuples\footnote{An $\mv x$-tuple is a tuple that assigns each variable in $\mv x$
a value in $D$.} with multiplicities.  Formally, for each
$\mv d \in D^{\mv x}$, denote
$Q(\calD)[\mv d] \defeq \setof{f \in \hom(Q,\calD)}{f(\mv x)=\mv d}$.
The {\em answer to $Q$ on $\calD$ under the bag-set semantics} is the
mapping $\mv d \mapsto |Q(\calD)[\mv d]|$.
The bag-set semantics corresponds to a $\texttt{count(*)-groupby}$
query in SQL.

Given two queries $Q_1, Q_2$ with the same number of head variables,
we say that {\em $Q_1$ is contained in $Q_2$ under bag-set semantics},
and denote with $Q_1 \preceq Q_2$, if for every $\calD$, we have
$Q_1(\calD) \leq Q_2(\calD)$, where $\leq$ compares functions
point-wise, $\forall \mv d, |Q_1(\calD)[\mv d]| \leq |Q_2(\calD)[\mv d]|$.

\begin{pbm}[Query containment problem under bag-set semantics] \label{problem:bcqc}
  Given $Q_1, Q_2$, check whether $Q_1 \preceq Q_2$.
\end{pbm}

A query $Q$ is called a {\em Boolean query} if it has no head
variables, $|\mv x|=0$.  It is known that the query containment problem
under bag semantics can be reduced to that of Boolean queries under
bag semantics.  For completeness, we provide the proof in
Appendix~\ref{appendix:problem}, and only mention here that the
reduction preserves all special properties discussed later in this
paper: acyclicity, chordality, simplicity.  For that reason, in this
paper we only consider Boolean queries, and denote
Problem~\ref{problem:bcqc} by $\bcqc$.


\paragraph{Bag-bag Semantics} In our setting the input database $\calD$
is a set, only the query's output is a bag.  This semantics is known
under the term {\em bag-set} semantics.  Query containment has also
been studied under the {\em bag-bag} semantics, where the database may
also have duplicates.  This problem is known to be reducible to the
containment problem under bag-set semantics \cite{DBLP:conf/pods/JayramKV06}, by adding a new attribute
to each relation, and for that reason we do not consider it further in
this paper.  One aspect of the bag-bag semantics is that repeated
atoms change the meaning of the query, while repeated atoms can be
eliminated under bag-set semantics.  For example
$R(x)\wedge R(x) \wedge S(x,y)$ and $R(x) \wedge S(x,y)$ are different
queries under bag-bag semantics, but represent the same query under
bag-set semantics.  Since we restrict to bag-set semantics we assume
no repeated atoms in the query.

\paragraph{The Domination Problem}
We briefly review two related problems that are equivalent to $\bcqc$.
Given two relational structures $\calA$ and $\calB$,
we say that $\calB$ {\em dominates} $\calA$, and write
$\calA \preceq \calB$, if $\forall \calD$,
$|\hom(\calA,\calD)| \leq |\hom(\calB,\calD)|$.

\begin{pbm}[The domination problem, $\dom$] Given a vocabulary $\calR$,
  and two structures $\calA, \calB$, check if $\calB$ dominates
  $\calA$: $\calA \preceq \calB$.
\end{pbm}

$\dom$ and $\bcqc$ are essentially the same problem.
Kopparty and Rossman~\cite{HDE} considered the following
generalization:

\begin{pbm}[The exponent-domination problem] Given a rational number
  $c \geq 0$ and two structures $\calA$ and $\calB$, check whether
  $|\hom(\calA,\calD)|^c \leq |\hom(\calB,\calD)|$ for all structures
  $\calD$.
\end{pbm}

This problem is equivalent to $\dom$, because it can be reduced to
$\dom$ by observing that
$|\hom(n\cdot \calA,\calD)| = |\hom(\calA, \calD)|^n$, where
$n\cdot\calA$ represents $n$ disjoint copies of $\calA$~\cite[Lemma
2.2]{HDE}.  Conversely, $\dom$ is the special case $c=1$.

%
%

\subsection{Information Inequality Problems}

In this paper all logarithms are in base $2$.
For a random variable $X$ with values that are in a finite domain $D$,
its (binary) {\em entropy} is defined by
\begin{align}
    H(X) &:= - \sum_{x \in D} \pr[X=x] \cdot \log \pr[X=x] \label{eqn:entropy}
\end{align}
Note that in the above definition, $X$ can be a tuple of random variables, in which case
$H(X)$ is their joint entropy. The entropy $H(X)$ is a non-negative real number.

Let $V=\set{X_1, \ldots, X_n}$ be a
set of $n$ random variables jointly distributed over finite domains.
For each $\alpha \subseteq [n]$, the joint distribution induces a marginal distribution
for the tuple of variables $X_\alpha = (X_i : i \in \alpha)$. One can also equivalently think of
$X_\alpha$ as a vector-valued random variable.
Either way, the marginal entropy on $X_\alpha$ is defined by~\eqref{eqn:entropy} too,
where we replace $X$ by $X_\alpha$.
Define the function $h : 2^{[n]} \rightarrow \R_+$ as $h(\alpha) \defeq H(X_\alpha)$, for all
$\alpha \subseteq [n]$.  We call $h$ an {\em entropic function} (associated with the joint
distribution on $V$) and
identify it with a vector $h \in \R^{2^n}_+$.

The set of all entropic
functions is denoted\footnote{Most texts drop the component
  $h(\emptyset)$, which is always 0, and define
  $\Gamma_n^* \subseteq \R^{2^n-1}_+$.  We prefer to keep the
  $\emptyset$-coordinate to simplify notations.} by
$\Gamma_n^* \subseteq \R^{2^n}_+$.  With some abuse, we blur the
distinction between the set $[n]$ and the set of variables
$V = \set{X_1, \ldots, X_n}$, and write $h(X_\alpha)$ instead of
$h(\alpha)$.

An {\em information inequality}, or $\ii$, defined by a vector
$c = (c_X)_{X\subseteq V}$ $\in \R^{2^V}$, is an inequality of the form
\begin{align}
    0 &\leq \sum_{X \subseteq V} c_X h(X)  \label{eq:iti:0}
\end{align}
The information inequality is {\em valid} if it holds for all $h \in \Gamma_n^*$~\cite{DBLP:journals/entropy/Chan11}.

\begin{pbm}[$\ii$-Problem] Given a set $V$ and a collection of integers $c_X$, for
  $X \subseteq V$, check whether the information
  inequality~\eqref{eq:iti:0} is valid.
\end{pbm}

A {\em max-information inequality}, or $\mii$, is defined by $k$
vectors $\mv c_{\ell} := (c_{\ell, X})_{X\subseteq V} \in \R^{2^V}$, $\ell \in [k]$,
and is written as:
\begin{align}
    0 & \leq \max_{\ell \in [k]} \sum_{X \subseteq V} c_{\ell,X} h(X)  \label{eq:miti:0}
\end{align}
The $\mii$ is {\em valid} if it holds for all entropic functions
$h \in \Gamma_n^*$.

\begin{pbm}[$\mii$ Problem] Given a set $V$ and integers $c_{\ell,X}$, for
  $\ell \in [k]$ and $X \subseteq V$, check whether the $\mii$
  \eqref{eq:miti:0} is valid.
\end{pbm}

We denote the $\ii$- and $\mii$ problems by $\iip$ and $\miip$
respectively.  Both are co-recursively enumerable
\conferenceorfull{}{(Appendix~\ref{sec:background:it:long})} and it is open if
any of them is decidable.


%
%
%

%% file: sec-main-results.tex
\section{Main Results}

\label{sec:main:results}

\subsection{Connecting $\bcqc$ to  Information Theory}

We state our first main result, and defer its proofs to
Sec.~\ref{sec:connection} and~\ref{sec:max-iti:completeness}.  Recall
that a {\em many-one reduction} of a decision problem $A$ to another
decision problem $B$, denoted by $A \leq_m B$, is a computable
function $f$ such that for every input $X$, the yes/no answer to
problem $A$ on $X$ is the same as the yes/no answer to the problem $B$
on $f(X)$.  This is a special case of a Turing reduction,
$A \leq_T B$, which means an algorithm that solves $A$ given access to
an oracle that solves $B$.  Two problems are {\em many-one
  equivalent}, denoted by $A \equiv_m B$, if $A \leq_m B$ and
$B \leq_m A$.
%
%
%
%

Our main result is that the $\miip$ is many-one equivalent to the
query containment problem under bag semantics, when the containing
query is restricted to be acyclic.  We briefly review acyclic queries
here (we only consider $\alpha$-acyclicity in this paper~\cite{DBLP:journals/jacm/Fagin83}):

\begin{defn} \label{def:td} A {\em tree decomposition} of a query $Q$
  is a pair $(T,\chi)$ where $T$ is an undirected forest\footnote{We
    allow $Q$ to be disconnected, in which case $T$ can be a forest,
    but we continue to call it a tree decomposition.} and
  $\chi : \nodes(T) \rightarrow 2^{\vars(Q)}$ satisfies (a) the running
  intersection property: $\forall x \in \vars(Q)$,\linebreak
  $\setof{t \in \nodes(T)}{x \in \chi(t)}$ is connected in $T$, and
  (b) the coverage property: for every $A \in \atoms(Q)$, there exists
  $t \in \nodes(T)$ s.t. $\vars(A) \subseteq \chi(t)$.  The sets
  $\chi(t)$ are called the {\em bags}\footnote{Not to be confused with
    the bag semantics.} of the tree decomposition. A query $Q$ is {\em
    acyclic} if there exists a tree decomposition $(T,\chi)$ such
  that, for all $t \in \nodes(T)$, $\chi(t) = \vars(A)$ for some
  $A \in \atoms(Q)$.
\end{defn}


\begin{thm} \label{th:main:result}
  Let $\bcqca$ denote the $\bcqc$ problem $Q_1 \preceq Q_2$, where
  $Q_2$ is restricted to acyclic queries.  Then
  $\miip \equiv_m \bcqca$.
\end{thm}

The proof of the theorem consists of three steps.  First, we describe
in Sec.~\ref{subsec:sufficient} a $\miip$ inequality that is
sufficient for containment, which is quite similar to, and inspired by
an inequality by Kopparty and Rossman~\cite{HDE}.  Second, we prove in
Sec.~\ref{subsec:necessary} that, when $Q_2$ is acyclic, then this
inequality is also necessary, thus solving the conjecture
in~\cite[Sec.3]{HDE}; our proof is based on Chan-Yeung's
group-characterizable entropic
functions~\cite{DBLP:conf/isit/Chan07,DBLP:journals/tit/ChanY02}.  In particular,
$\bcqca \leq_m \miip$.  We do not know if this can be strengthened to
$\bcqc$ and/or $\iip$ respectively.  Finally, we give the many-one
reduction $\miip \leq_m \bcqca$ in
Sec.~\ref{sec:max-iti:completeness}.


\subsection{Novel Decidable Class of $\bcqc$}

\label{subsec:results:decidable}

Our next two results consist of a novel decidable class of query
containment under bag semantics, and, correspondingly, a novel
decidable class of max-information inequalities.  We state here the
results, and defer their proofs to
Section~\ref{sec:decidability}.

We show  that containment is decidable when
$Q_2$ is {\em chordal} and admits a {\em simple} junction tree (decomposition);
to formally state the result, we define chordality, simplicity, and junction tree next.

A query $Q$ is said to be {\em chordal} if its Gaifman graph $G$ is chordal, i.e., there is
a tree decomposition of $G$ in which every bag induces a clique of $G$.
A tree decomposition of $G$ (and thus of $Q$)
where all bags induce {\em maximal cliques} of $G$ is called a {\em junction tree} in the
graphical models literature (see Def. $2.1$ in~\cite{DBLP:journals/ftml/WainwrightJ08}).

Fix a tree decomposition of a query $Q$, and let $t \in \nodes(T)$.
A tree decomposition is called {\em simple} if
$\forall (t_1,t_2) \in \edges(T)$, $|\chi(t_1)\cap \chi(t_2)| \leq 1$,
and is called {\em totally disconnected} if\footnote{Equivalently,
  $\edges(T)=\emptyset$, because any edge
  s.t. $\chi(t_1)\cap\chi(t_2)=\emptyset$ can be removed.}
$\forall (t_1,t_2) \in \edges(T)$,
$\chi(t_1)\cap \chi(t_2) = \emptyset$.
As an example of a totally disconnected tree decomposition, consider the query
$Q() \leftarrow R(a), S(b)$ and a tree decomposition of $Q$ with only two nodes $t_1$ and $t_2$
where $\chi(t_1) =\{a\}$ and $\chi(t_2)=\{b\}$.

Note that every acyclic query is chordal, but not necessarily simple;
for example, the query $Q() \leftarrow R(a,b,c), S(b,c,e)$
is a non-simple acyclic query.
Conversely a chordal query is not necessarily acyclic; for example, any $k$-clique query
with $k \geq 3$ is chordal.

\begin{thm} \label{th:decidable} Checking $Q_1 \preceq Q_2$ is
  decidable in exponential time when $Q_2$ is chordal and admits a simple
  junction tree.
\end{thm}

Next, we  complement Theorem~\ref{th:decidable} by showing
that, if $Q_1 \not\preceq Q_2$ then there exists a ``witness'' with a
simple structure.  This result is similar in spirit to other results
where a decision problem can be restricted to special databases: for
example, query containment under set semantics holds iff it holds on
the canonical database of $Q_1$~\cite{DBLP:conf/stoc/ChandraM77}, and
implication between functional dependencies holds iff it holds on all
relations with two tuples.

Let $Q_1$ be a query and $V = \vars(Q_1)$.  A relation
$P \subseteq D^{V}$ is called a {\em $V$-relation}.
A $V$-relation $P$ and $Q_1$ {\em induce} a database instance
$\Pi_{Q_1}(P) \defeq (D, R_1^D, \ldots, R_m^D)$ where,
\begin{eqnarray}
    \forall \ell \in [m]: & R_\ell^D & \defeq \bigcup_{A \in \atoms(Q_1): \rel(A)=R_\ell} \Pi_{\vars(A)}(P)
\label{eq:projection:db}
\end{eqnarray}
In other words, we project $P$ on each atom, and define $R_\ell^D$ as
the union of projections on atoms with relation name $R_\ell$.

The notation $\Pi_{\vars(A)}(P)$ requires some explanation, because
the atom $A$ may have repeated variables, thus $\vars(A)$ is a
function (described in~\eqref{eqn:vars:as:functions}).  Given a set of
integer indices $Y$ and a function $\varphi : Y \rightarrow {V}$, the
{\em generalized projection} is
$\Pi_\varphi(P) \defeq \setof{f \circ \varphi}{f \in D^{V}}$.  A tuple
$f \in D^V$ is a function $V \rightarrow D$, hence $f \circ \varphi$
just denotes function composition.  For example, if $Q_1 = R(x,x,y)$
and $P = \set{(a,b)}$, then $R^D = \Pi_{(x,x,y)}(P) = \set{(a,a,b)}$.
Obviously $P \subseteq \hom(Q_1,\Pi_{Q_1}(P))$, which means
$|P| \subseteq |\hom(Q_1,\Pi_{Q_1}(P))|$, and this implies:

\begin{fact}[Witness]
If there exists a $\vars(Q_1)$-relation $P$ such that
  $|P| > |\hom(Q_2, \Pi_{Q_1}(P))|$, then $Q_1 \not\preceq Q_2$,
  in which case $P$ is said to be a {\em witness} (for the fact that $Q_1 \not\preceq Q_2$).
\end{fact}


We next define two special types of relations (and witnesses) that have interesting analogues in
information theory and thus arise naturally when doing reductions between the database
world and the information theory world.
Let $W$ be a set of integer indices.
Fix $\psi : W \rightarrow 2^{V}$ and a tuple $f \in D^{V}$.
For any index $y \in W$, we view $f(\psi(y))$ as an atomic value in the
domain $D^{\psi(y)}$.  Define the $W$-tuple
$\psi \cdot f \defeq (f(\psi(y)))_{y \in W}$; its components may
belong to different domains.

\begin{defn}[Product and normal relations] \label{def:normal:relation}
  A $V$-relation $P$ is a {\em product relation} if
  $P = \prod_{x \in V} S_x$, where each $S_x$ is a unary relation.
  A $W$-relation is called a {\em normal relation} if it is of the form
  $\setof{\psi \cdot f}{f \in P}$ where $P$ is some product $V$-relation and
  $\psi: W \rightarrow 2^{V}$ is some function.
\end{defn}

One can verify that every product relation is a normal relation.
For a simple illustration, consider the case when $V = \{X_1,X_2\}$.
A product relation on $V$ is
$\setof{(u,v)}{u, v \in [N]} =[N] \times [N]$.  A normal relation with
four attributes is $\setof{(uv,u,v,v)}{u, v \in [N]}$, where $uv$
denotes the concatenation of $u$ and $v$.
This normal relation corresponds to the map $\psi : [4] \to 2^V$ where
$\psi(1) = \{X_1, X_2\}$,
$\psi(2) = \{X_1\}$, and
$\psi(3) = \psi(4) = \{X_2\}$.
In a product relation all
attributes are independent, while a normal relation may have
dependencies: in our example the first attribute $uv$ is a key, and the
last two attributes are equal.

\begin{thm} \label{th:product:normal:databases} Let $Q_2$ be chordal,
  \begin{itemize}
      \item[(i)] If $Q_2$ admits a totally disconnected junction tree,
          then $Q_1 \not\preceq Q_2$ if and only if there is a product witness.
      \item[(ii)] If $Q_2$ admits a simple junction tree, then $Q_1 \not\preceq Q_2$
    if and only if there exists a normal witness.
  \end{itemize}
\end{thm}

We prove both theorems in Section~\ref{sec:decidability}, using
the novel results on information-theoretic inequalities described
next, in Section~\ref{sec:background:it:short}.

\begin{ex} \label{ex:normal:database} We illustrate with the following
  queries:
{\small
  \begin{align*}
    Q_1 = & A(x_1,x_2) \wedge B(x_1,x_2) \wedge C(x_1,x_2)\wedge  A(x_1',x_2') \wedge B(x_1',x_2') \wedge C(x_1',x_2')\\
    Q_2 = & A(y_1,y_2) \wedge B(y_1,y_3) \wedge C(y_4,y_2)
  \end{align*}
}
$Q_2$ is acyclic with a simple junction tree:
$\set{y_1,y_3} - \set{y_1,y_2} - \set{y_2,y_4}$.  We prove that
$Q_1 \not\preceq Q_2$ has a normal witness:
  \begin{align*}
    P \defeq &  \setof{(u,u,v,v)}{u\in  [n], v \in [n]} \subseteq D^{\set{x_1,x_2,x_1',x_2'}}
  \end{align*}
  $P$ induces the database $\Pi_{Q_1}(P) = ([n], A^D,B^D, C^D)$, where
  $A^D = B^D = C^D = \setof{(u,u)}{u \in [n]}$, and
  $|P| = n^2 > |\hom(Q_2,\Pi_{Q_1}(P))| = n$ when $n > 1$, proving
  $Q_1 \not\preceq Q_2$.

  On the other hand, there is no product relation $P$ that can witness
  $Q_1 \not\preceq Q_2$.  Indeed, if
  $P = S_1 \times S_2 \times S_3 \times S_4$ where $S_1, \ldots, S_4$
  are unary relations, then the associated database $\Pi_{Q_1}(P)$ has
  relations
  $A^D = B^D = C^D \defeq (S_1\times S_2) \cup (S_3 \times S_4)$, and
  therefore
  $|\hom(Q_2, \Pi_{Q_1}(P))| \geq \max(|S_1 \times S_2|^2,|S_3 \times
  S_4|^2) \geq |S_1 \times S_2 \times S_3 \times S_4| = |P|$.
\end{ex}

\subsection{Novel Class of Shannon-Inequalities}

\label{sec:background:it:short}

Our decidability results are based on a new result on
information-theoretic inequalities, proving that certain max-linear
inequalities are essentially Shannon inequalities.  To present it, we
need to review some known facts about entropic functions.  We refer to
Appendix~\ref{sec:background:it:long} and to~\cite{Yeung:2008:ITN:1457455} for additional information.  Recall that the set of entropic functions over $n$
variables is denoted $\Gamma_n^* \subseteq \R^{2^n}$, and that we blur
the distinction between a set $V$ of $n$ variables  and $[n]$.

\input{table}

We begin by discussing closure properties of entropic functions and then introduce certain special classes of entropic functions.
For the benefit of the readers familiar
with database theory, we give in Table~\ref{tab:DB:IT} the mapping
between some of the database concepts used in this paper and their
information-theoretic counterparts.
For our discussion, it is useful to define the notion of the {\em entropy of a relation}.
Given a $V$-relation $P$, its \emph{entropy} is the entropy of the joint distribution on
$V$, uniform on the support of  $P$ (i.e., tuples in $P$).

First, the sum of two entropic functions is also an entropic function, that is, if $h_1, h_2 \in \Gamma_n^*$, then $h_1+h_2 \in \Gamma_n^*$. It follows that if $k$ is a positive integer and $h$ is an entropic function, then the function $h'=kh$ is also entropic. However, if $c>0$ is a positive real number and $h$ is an entropic function, then the function $h'=ch$ need not be entropic, in general.
In contrast, the function $h'=ch$ is entropic, if $c>0$ is a positive real number and $h$ is a \emph{step function}, defined as follows.
Let
$W\subsetneq V$ be a proper subset of $V$. The \emph{step function at $W$}, denoted by $h_W$, is the function
\begin{align*}
  h_W(X) = &
  \begin{cases}
    0 & \mbox{if $X \subseteq W$} \\
    1 & \mbox{otherwise.}
  \end{cases}
\end{align*}
Every step function $h_W$ is entropic.  To see this, consider the relation $P_W =  \set{f_1, f_2} \subseteq \{1,2\}^{V}$, where
 $f_1 = (1,1,\ldots,1)$ and $f_2  = (\underbrace{2,\ldots,2}_{V - W},\underbrace{1,\ldots,1}_W)$, that is, $f_2$ has $1$'s on the positions $W$
and $2$'s on all  other positions.
It is not hard to verify that $h_W$ is the entropy of the relation $P_W$,
and thus the step function $h_W$ is indeed entropic.

As mentioned above, if $c>0$ is a positive real number and $h_W$ is a step function, then the function $h'=ch_W$ is entropic; the proof of this fact is given in Appendix~\ref{sec:background:it:long}.
A {\em normal entropic} function, or simply \emph{normal} function, is a
non-negative linear combination of step functions, i.e., $\sum_{W \subsetneq V} c_W h_W$, for $c_W \geq 0$.
We write $\calN_n$ to denote the set of all normal  functions.
Since, as mentioned earlier, the sum of two entropic functions is entropic, it follows that every normal  function is entropic; thus, we have that  $\calN_n \subseteq \Gamma_n^*$.
In Appendix~\ref{sec:background:it:long},
we show that the  normal functions  are
precisely the entropic functions with a non-negative I-measure
(defined by Yeung~\cite{Yeung:2008:ITN:1457455}).  The term ``normal''
was introduced in~\cite{DBLP:conf/pods/KhamisNS16}.  One can check
that the entropy of every normal relation (Def.~\ref{def:normal:relation}) is a normal function.

\begin{ex} \label{ex:parity} The {\em parity function} is the entropy
  of the following relation with 3 variables:
  $P = \setof{(X,Y,Z)}{X, Y, Z \in \set{0,1}, X\oplus Y \oplus Z =
    0}$ where $\oplus$ is the exclusive OR.  More precisely, the entropy is $h(X)=h(Y)=h(Z)=1$,
  $h(XY)=h(XZ)=h(YZ)=h(XYZ)=2$.  We show in Sec.~\ref{sec:domination}
  that $h$ is not normal.
\end{ex}

A function $h : 2^{V} \rightarrow \R_+$ is called {\em modular} if
it satisfies $h(X \cup Y) + h(X \cap Y) = h(X) + h(Y)$ for all
$X,Y \subseteq V$, and $h(\emptyset)=0$. It is easy to show that $h$ is modular iff
$h(X_\alpha) = \sum_{i\in\alpha} h(X_i)$ for all $\alpha\subseteq V$.
 It is immediate to check that the
entropy of any product relation (Def.~\ref{def:normal:relation}) is
modular.
  We write $\calM_n$ to denote the set of all modular functions.
Every modular function is normal, hence it is also entropic.
To see this, given a modular function $h$, for each $i\leq n$,
define $W_i= V \setminus \{X_i\}$ and let $h_{W_i}$ be the associated step function at
$W_i$. It is now easy to verify that $h= \sum_{i=1}^n h(X_i) \cdot h_{W_i}$, thus $h$ is a normal function. In summary, we have $\calM_n \subseteq \calN_n \subseteq \Gamma_n^*$.

All entropic functions satisfy Shannon's \emph{basic} inequalities, called
{\em monotonicity} and {\em submodularity},
\begin{align}
    h(X) & \leq h(X\cup Y)  & h(X \cup Y) + h(X \cap Y) & \leq h(X)+h(Y) \label{eq:shannon}
\end{align}
for all $X, Y\subseteq V$. (Since $h(\emptyset) = 0$, monotonicity implies {\em
non-negativity} too.)
 A function
$h : 2^{V} \rightarrow \R_+$, $h(\emptyset)=0$, that satisfies
Eq.(\ref{eq:shannon}) is called a {\em polymatroid}, and the set of all
polymatroids is denoted by $\Gamma_n$.  Thus, $\Gamma^*_n \subseteq \Gamma_n$.
Zhang and Yeung \cite{DBLP:journals/tit/ZhangY97} showed that $\Gamma^*_n$ is properly contained in $\Gamma_n$, for every $n\geq 4$.
Any inequality derived by taking a non-negative linear combination of inequalities
\eqref{eq:shannon} is called a {\em Shannon inequality}.
In a follow-up paper \cite{zhang1998characterization}, Zhang and Yeung gave the first example of a
$4$-variable valid information inequality which is non-Shannon.

In summary, we have considered the chain of the following
four sets:
$\calM_n \subsetneq \calN_n \subsetneq \Gamma_n^* \subsetneq
\Gamma_n$.  Except for $\Gamma_n^*$, each of these sets is a polyhedral cone.
Using basic linear programming, one can show that it is decidable whether a max-linear
inequality holds on a polyhedral set. In contrast, (even) the topological closure of $\Gamma_n^*$ is not
polyhedral~\cite{DBLP:conf/isit/Matus07}; in fact,  it is conjectured to not even be
semi-algebraic~\cite{DBLP:journals/ijicot/GomezCM17}, and it is an
open question whether linear inequalities or max-linear inequalities on $\overline
\Gamma^*_n$ are decidable.

For a given vector $(c_X)_{X \subseteq V} \subseteq \R^{2^n}$
where $c_\emptyset=0$, we associate a {\em linear expression} $E$ which is the linear function
$E(h) \defeq \sum_{X \subseteq V} c_X h(X)$. As stated earlier,  a linear inequality
$E(h) \geq 0$ that is valid for all $h \in \Gamma_n^*$ is called an
{\em information inequality}; furthermore, a {\em
max information inequality} is one of the form $\max_\ell E_\ell(h) \geq 0$,
where $\forall \ell, E_\ell$ is a linear expression.

In this paper, for any variable sets $X,Y \subseteq V$, we write $h(XY)$ as a shorthand for
$h(X \cup Y)$, and define the {\em conditional entropy} to be
$h(Y|X) \defeq h(XY) - h(X)$.  Despite its name, the mapping
$Y \mapsto h(Y|X)$ is not always an entropic function
(Appendix~\ref{sec:background:it:long}), but it is always a limit of entropic
functions.
The submodularity law~\eqref{eq:shannon} can be written using conditional entropies as
\begin{align}
    h(XY|X) &\leq h(Y|X\cap Y) \label{eqn:submod}
\end{align}

\bdefn[Simple and unconditioned linear expressions]
We  call the term $h(Y|X)$ {\em simple} if
$|X|\leq 1$. A simple term $h(Y|X)$ is {\em unconditioned} if $X = \emptyset$.
A {\em conditional linear expression} is a linear expression $E$ of the form
$E(h) = \sum_{X\subseteq Y \subseteq V} d_{Y|X} \cdot h(Y|X)$, where $d_{Y|X}$ are
non-negative coefficients.
A conditional linear expression is said to be {\em simple} (respectively, unconditioned)
if $d_{Y|X}>0$ implies $h(Y|X)$ is simple (respectively, unconditioned).
\label{defn:simple-unconditioned}
\edefn

\bdefn[Decidable classes of inequalities]
A class $\calI$ of inequalities over variables
$h : 2^{[n]} \rightarrow \R_+$ is {\em decidable} if the problem of determining whether
a given inequality $I \in \calI$ holds for all $h\in\Gamma_n^*$ is decidable.
\label{defn:decidable-class}
\edefn

\bdefn[Essentially Shannon inequalities]
Let $\calI$ be a class of max-linear inequalities.  We say that
$\calI$ is {\em essentially Shannon} if, for every inequality $I$ in
$\calI$, $I$ holds for every $h \in \Gamma_n^*$ if and only if $I$ holds for every
$h \in \Gamma_n$. Any essentially Shannon class is decidable, because $\Gamma_n$ is
polyhedral.
\label{defn:essentially-shannon}
\edefn


\begin{thm} \label{th:simple} Consider a
  max-linear inequality of the following form, where $q > 0$, and $E_\ell$ are conditional
  linear expressions:
  \begin{align}
      q\cdot h(V) \leq & \max_{\ell\in [k]} E_\ell(h) \label{eq:special:case}
  \end{align}
  \begin{itemize}
      \item[(i)] Suppose that $E_\ell$ is unconditioned, $\forall \ell \in [k]$;  then
    inequality~\eqref{eq:special:case} holds $\forall h \in \calM_n$ if and only if
    it holds $\forall h \in \Gamma_n$.
\item[(ii)] Suppose that $E_\ell$ is simple, $\forall \ell \in [k]$; then, inequality
    \eqref{eq:special:case} holds $\forall h \in \calN_n$ if and only if it holds
    $\forall h \in \Gamma_n$.
  \end{itemize}
  In particular, the class of inequalities (\ref{eq:special:case}),
  where each $E_\ell$ is simple, is essentially Shannon and decidable.
  (Recall Definitions~\ref{defn:simple-unconditioned},~\ref{defn:decidable-class} and~\ref{defn:essentially-shannon}.)
\end{thm}

The proof of the theorem follows from  a technical lemma, which is of
independent interest:

\begin{lmm} \label{lemma:h:domination} Let
  $h : 2^{[n]} \rightarrow \R_+$ be any polymatroid.  Then there
  exists a normal polymatroid $h' \in \calN_n$ with the following
  properties:
  \begin{enumerate}
  \item \label{item:domination:1} $h'(X) \leq h(X)$, forall $X \subseteq [n]$,
  \item \label{item:domination:2} $h'([n]) = h([n])$,
  \item \label{item:domination:3} $h'(\set{i}) = h(\set{i})$, forall $i \in [n]$.
  \end{enumerate}
  In addition, there exists a modular function $h''\in \calM_n$ that
  satisfies conditions (\ref{item:domination:1}) and
  (\ref{item:domination:2}).
\end{lmm}

%
%

This lemma says that every polymatroid $h$ can be decreased to become a normal
polymatroid $h'$, while preserving the values at $[n]$ (all variables)
and at all singletons $\set{i}$.  If we drop the last condition, then
the existence of a modular function $h''$ follows from the
modularization lemma~\cite{DBLP:conf/pods/Khamis0S17}, which is based
on Lovasz's monotonization of submodular functions:
\begin{align*}
  h''(X) \defeq \sum_{i \in X} h(\set{i}|[i-1])
\end{align*}
The proof that one can also satisfy condition
(\ref{item:domination:3}), by relaxing from a modular function to a normal
one, is non-trivial and given in Sec.~\ref{sec:domination}.

\begin{proof}[Proof of Theorem~\ref{th:simple}] We prove the second
  item.  Let $E(h) \defeq \max_\ell E_\ell(h) - q\cdot h(V)$, where
  each $E_\ell$ has the form $\sum_i h(Y_i|X_i)$ with $|X_i| \leq 1$.
  Let $h \in \Gamma_n$, and let $h' \in \calN_n$ be the normal
  polymatroid in Lemma~\ref{lemma:h:domination}. For every $\ell$, we
  have
  $E_\ell(h') = \sum_i h'(X_iY_i) - \sum_i h'(X_i) \leq \sum_i
  h(X_iY_i) - \sum_i h(X_i) = E_\ell(h)$, because $|X_i|\leq 1$ and
  therefore $h'(X_i)=h(X_i)$.  Since $E(h')\geq 0$, we obtain
  $q \cdot h(V) = q \cdot h'(V) \leq \max_\ell E_\ell(h') \leq
  \max_\ell E_\ell(h)$ completing the proof. The first item of the
  theorem is proven similarly, and omitted.
\end{proof}

\begin{ex} \label{eq:e1e2e3} We illustrate Theorem~\ref{th:simple} here with an inequality
  needed later in Ex.~\ref{ex:vee}.  Consider
  $h(X_1X_2X_3) \leq \max(E_1,E_2,E_3)$, where:
  \begin{align*}
    E_1 &= h(X_1X_2)+h(X_2|X_1) \\
    E_2 &= h(X_2X_3)+h(X_3|X_2) \\
    E_3 &= h(X_1X_3)+h(X_1|X_3)
  \end{align*}
  Notice that all three expressions are simple, hence part (ii) of the theorem
  applies.
  In particular according to Theorem~\ref{th:simple}, in order to check
   whether the above inequality holds for all entropic $h\in\Gamma^*_3\supseteq \calN_3$, it is sufficient
   to check that the inequality holds for all polymatroids $h\in\Gamma_3$. (This latter check is much easier
   than the former because $\Gamma_3$ is polyhedral while $\Gamma^*_3$ is not.
   The non-trivial direction of the theorem is proving that if the inequality fails on some
   $h \in \Gamma_3$, then it must fail on some $h' \in \calN_3\subseteq\Gamma^*_3$.)
   In this example, it turns out the above inequality does indeed hold for all $h\in\Gamma_3$.
  In particular, using Shannon's submodularity law~\eqref{eqn:submod}, we infer
  $E_1 = h(X_1X_2)+h(X_2|X_1) \geq h(X_1X_2)+h(X_2|X_1X_3)$
  and, similarly for $E_2$ and $E_3$; therefore,
\begin{eqnarray*}
    \max(E_1,E_2,E_3) &\geq&  \frac 1 3 [E_1+E_2+E_3]\\
   &\geq&  \frac 1 3 \bigl[h(X_1X_2)+h(X_2|X_1X_3)
   + h(X_2X_3)+h(X_3|X_1X_2)
   + h(X_1X_3)+h(X_1|X_2X_3)\bigr]\\
   &=&   h(X_1X_2X_3).
\end{eqnarray*}
\end{ex}


%% file: table.tex
\begin{table*}
\begin{tabular}{|C{0.47\textwidth}|C{0.47\textwidth}|}
\hline
Database Theory & Information Theory \\
\hline\hline
$P\subseteq D^{ V}$\newline
A {\em relation} $P$ over a set of $n$ variables $ V$,
each of which has domain $D$&
$h\in\Gamma^*_n$\newline
An {\em entropic function} $h:2^{ V}\rightarrow \R_+$ over a set of $n$ variables $ V$.\newline
$h$ is defined by a uniform probability distribution $p$ over $P$.\\
\hline
$P = S_1\times\cdots\times S_n\subseteq D^{ V}$\newline
A {\em product relation} $P$ (Definition~\ref{def:normal:relation})&
$h(X) = \sum_{i\in X} h(i), \quad\text{forall $X \subseteq  V$}$\newline
A {\em modular function} $h\in\calM_n$\\
\hline
The set of product relations&
The set of modular functions $\calM_n$
\\
\hline
$P=P_1\otimes P_2$, \quad where $P_1\subseteq D_1^{ V}, P_2\subseteq D_2^{ V}, P \subseteq (D_1\times D_2)^{ V}$\newline
A {\em domain product} $P$ of two relations $P_1, P_2$, all of which are over the same variable set $ V$ (Definition~\ref{def:domain:product})&
$h = h_1+h_2, \quad\text{where $h, h_1, h_2\in \Gamma^*_n$}$\newline
A {\em sum} $h$ of two entropic functions $h_1, h_2$, all of which are over $n$ variables\\
\hline
$P_W\defeq \set{f_1, f_2}\subseteq D^{ V}$, $\text{for some $W \subseteq  V$, where}$\newline
$\begin{array}{lcl}
    f_1 &\defeq& (1,1,\ldots,1),\\
    f_2 &\defeq& (\underbrace{2,\ldots,2}_{ V - W}, \underbrace{1,\ldots,1}_W),
\end{array}$\newline
Given $W \subseteq  V$, the relation $P_W$ has two tuples $f_1, f_2$ differing only in positions $ V-W$. (See Section~\ref{sec:background:it:short})
&
$  h_W(X) \defeq
\begin{cases}
0 & \mbox{if $X \subseteq W$} \\
1 & \mbox{otherwise}
\end{cases}$\newline
Given $W\subseteq  V$, a {\em step function} $h_W$.
\\
\hline
$P = P_{W_1}\otimes P_{W_2}\otimes\cdots\otimes P_{W_m}$\newline
A {\em normal relation} $P$ over variable set $ V$ is a domain product of $m$ (not necessarily distinct) relations $P_{W_i}$ for $W_i\subseteq  V$\newline
(Another way to phrase Definition~\ref{def:normal:relation}) &
$\displaystyle{h = \sum_{W\subseteq  V} c_{W} h_W, \quad\text{where $c_W\geq 0$}}$\newline
A {\em normal entropy} $h\in\calN_n$ is a non-negative weighted sum of step functions $h_W$\\
\hline
The set of normal relations&
The set of normal functions $\calN_n$ $\equiv$ \newline
the cone closure of step functions\\
\hline
$P_W$, when $| V-W|=1$, becomes a product relation&
$h_W$, when $| V-W|=1$, becomes a modular function\\
\hline
Product relations are a proper subclass of normal relations&
Modular functions are a proper subclass of normal functions\newline
$\calM_n \subsetneq \calN_n$
\\
\hline
A {\em group-characterizable} relation~\cite{DBLP:conf/isit/Chan07} $P \defeq \setof{(aG_1, \ldots, aG_n)}{a \in G}$, where $G$ is a group and $G_1, \ldots, G_n$ are subgroups
&
An entropic function $h\in\Gamma^*_n$
\\
\hline
The set of group-characterizable relations&
$\Gamma^*_n$\\
\hline
--&
$\Gamma_n - \Gamma^*_n$\newline
Polymatroids that are not entropic have no analog in databases
\\
\hline
\end{tabular}
\caption{Translation between the database world and the information theory world.}
\label{tab:DB:IT}
\end{table*}

%% file: sec-sufficient-and-necessary.tex
\section{Reducing $\bcqca$ to $\miip$}
\label{sec:connection}

In this section, we prove that $\bcqca \leq_m \miip$, showing half of the
equivalence claimed in Theorem~\ref{th:main:result}.
We start by associating to each query
containment problem a max-information inequality.  We then prove, two
results: the inequality is always a sufficient condition for
containment, and it is also necessary when the containing query is
acyclic.
From now on, we will use only upper case to denote variables, both
random variables and query variables.

Before we begin, we need to introduce some notations.  Fix a relation
$P \subseteq D^{ V}$ and a probability distribution with mass function
$p : P \rightarrow [0,1]$.
If $X \subseteq { V}$ is a set of variables, and
$\varphi : Y \rightarrow  V$ is a function, then recall that
$\Pi_X(P)$ and $\Pi_\varphi(P)$ denote the standard, and the
generalized projections respectively.  We write $\Pi_X(p)$ for the
standard $X$-marginal of $p$, and write $\Pi_\varphi(p)$ for the {\em
  $\varphi$-pullback}\footnote{This is a generalization of the
  pullback in~\cite[Sec.4]{HDE}.}.  The latter is a probability
distribution on $\Pi_\varphi(P)$ defined as follows.  Start from the
standard marginal $\Pi_{\varphi(Y)}(p)$ on $\Pi_{\varphi(Y)}(P)$, then
apply the isomorphism $\Pi_\varphi(P) \rightarrow \Pi_{\varphi(Y)}(P)$
defined as $\Pi_\varphi(f) \mapsto \Pi_{\varphi(Y)}(f)$,
$\forall f \in P$.  Finally, if $E = \sum_i c_i h(Y_i)$ is a linear
expression of entropic terms, where each $Y_i \subseteq Y$, then we
denote by $E \circ \varphi \defeq \sum_i c_i h(\varphi(Y_i))$ the
result of applying the substitution $\varphi$ to each term in $E$.

\begin{ex} \label{ex:projection} Let $ V=\set{X_1,X_2,X_3}$,
  $P \subseteq D^{ V}$,
  $\varphi(Y_1)=X_1, \varphi(Y_2)=\varphi(Y_3)=X_2$.  The generalized
  projection is
  $\Pi_\varphi(P) = \setof{(a,b,b)}{(a,b,c)\in P}\subseteq
  D^{\set{Y_1,Y_2,Y_3}}$.  Its tuples are in 1-1 correspondence with
  the standard projection
  $\Pi_{\varphi(Y)}(P) = \Pi_{X_1X_2}(P) = \setof{(a,b)}{(a,b,c) \in
    P}$.  If $p$ is a distribution on $P$, then the $\varphi$-pullback
  is
  $\Pi_\varphi(p)(Y_1Y_2Y_3=abb) \defeq p(X_1X_2=ab) = \sum_{c}
  p(X_1X_2X_3=abc)$.  Notice that we do not need to define the
  pullback for $(a,b,c)$ where $b\neq c$, because
  $(a,b,c)\not\in \Pi_\varphi(P)$.  Consider now the linear expression
  $E = 3h(Y_1) + 4h(Y_2Y_3) -6h(Y_3)$. Then
  $E\circ \varphi = 3h(X_1)+4h(X_2)-6h(X_2)=3h(X_1)-2h(X_2)$.
\end{ex}

We will introduce now a fundamental expression, $E_T$, that connects
query containment to information inequalities; we discuss its history
in Sec.~\ref{sec:discussion}.  Fix a tree decomposition $(T,\chi)$ of
some query $Q$, and recall that $T$ may be a forest.  Choose a root
node in each connected component, thus giving an orientation of $T$'s
edges, where each node $t$ has a unique $\parent(t)$.  We associate to
$T$ the following linear expression of entropic terms:
\begin{align}
    E_{(T,\chi)}(h) & \defeq \sum_{t \in \nodes(T)} h(\chi(t) | \chi(t) \cap \chi(\parent(t)))
\label{eq:et}
\end{align}
where $\chi(\parent(t))=\emptyset$ when $t$ is a root node.  We
abbreviate $E_{(T,\chi)}$ with $E_T$ when $\chi$ is clear from the
context. Expression~\eqref{eq:et} is independent of the choice of the root
nodes, because one can check that
$E_T = \sum_{t \in \nodes(T)} h(\chi(t)) - \sum_{(t_1,t_2) \in
  \edges(T)} h(\chi(t_1)\cap \chi(t_2))$.

\subsection{A Sufficient Condition}

\label{subsec:sufficient}


Henceforth, let $\td(Q)$ denote the set of all tree decompositions of a given query $Q$.

\begin{thm} \label{th:sufficient} Let
  $Q_1, Q_2$ be two conjunctive queries, $n=|\vars(Q_1)|$.
  If the following $\mii$ inequality holds $\forall h \in \Gamma_n^*$:
   \begin{align}
      h(\vars(Q_1)) &\leq  \max_{(T,\chi) \in \td(Q_2)} \max_{\varphi \in \hom(Q_2,Q_1)} (E_T\circ\varphi)(h)  \label{eqn:another:sufficient}
   \end{align}
   then $Q_1 \preceq Q_2$.
\end{thm}

The theorem is inspired by, and is similar to Theorem 3.1 by Kopparty and
Rossman~\cite{HDE}, with three differences. First, the result
in~\cite{HDE} applies only to graphs (i.e., databases with a single
binary relation symbol), while our result applies to arbitrary
relational schemas.  Second, we do not restrict $Q_2$ to be chordal.
Finally, \cite{HDE} restrict $h$ to entropies satisfying the
independence constraints defined by $Q_1$; while this restriction is
not needed to prove Theorem~\ref{th:sufficient}, it was needed
in~\cite{HDE} to prove necessity in a special case (Theorem 3.3
in~\cite{HDE}).  We will prove necessity in Theorem~\ref{th:necessary} in the next section without
needing this restriction.
Our proof of Theorem~\ref{th:sufficient} in this section is an
extension of the proof in~\cite{HDE}.
The proofs of both Theorem~\ref{th:sufficient} and \ref{th:necessary}
 use the following notation.  Give a node $t\in \nodes(T)$ of
tree decomposition of $Q$, we denote by $Q_t$ the
``subquery at $t$'', consisting of all atoms $A \in \atoms(Q)$ s.t.
$\vars(A) \subseteq \chi(t)$.  We can assume
w.l.o.g. (Appendix~\ref{appendix:problem}) that $\vars(Q_t) = \chi(t)$.
Before we present our proof of Theorem~\ref{th:sufficient},
we give an example, also from~\cite{HDE}, that illustrates the main idea of the proof.

\begin{ex} \label{ex:vee} This example is attributed to Eric Vee
  in~\cite{HDE}:
  \begin{eqnarray*}
    Q_1 &=& R(X_1,X_2)\wedge R(X_2,X_3) \wedge R(X_3,X_1),\\
    Q_2 &=& R(Y_1,Y_2)\wedge R(Y_1,Y_3).
  \end{eqnarray*}
We show that $Q_1 \preceq Q_2$.  Query $Q_2$ is acyclic, and its tree
decomposition $T$ is $\set{Y_1,Y_2} - \set{Y_1,Y_3}$, therefore:
\begin{align*}
    E_T &= h(Y_1Y_2) + h(Y_3|Y_1) = h(Y_1Y_2)+h(Y_1Y_3)-h(Y_1)
\end{align*}
There are three homomorphisms $\varphi : Q_2 \rightarrow Q_1$, hence
inequality (\ref{eqn:another:sufficient}) becomes:
\begin{align}
  h(X_1X_2X_3) & \leq \max(E_1,E_2,E_3) \label{ineq:vee}
\end{align}
where $E_1,E_2,E_3$ are the linear expressions in
Example~\ref{eq:e1e2e3}, where we have shown that the inequality holds
forall entropic $h$.  Theorem~\ref{th:sufficient} implies
$Q_1 \preceq Q_2$.  Here we prove the theorem
on this particular example.  Consider any database $\calD$, let
$P_1=\hom(Q_1,\calD)$, $p_1$ the uniform probability space on $P_1$,
and $h_1$ its entropy.  Since $h_1$ satisfies inequality
(\ref{ineq:vee}), one of the three terms on the right is larger than
the left, assume w.l.o.g. that this term corresponds to the
homomorphism $\varphi(Y_1)=X_1$,
$\varphi(Y_2)=\varphi(X_3)=X_2$. Thus,
$h_1(X_1X_2X_3) \leq h_1(X_1X_2) + h_1(X_2|X_1)$.  Let
$P_2 = \hom(Q_2,\calD)$. This is a relation with attributes
$Y_1,Y_2,Y_3$.  We define a probability distribution $p_2$ on $P_2$ as
follows: the marginal $p_2(Y_1,Y_2)$ is the same as $p_1(X_1,X_2)$,
and the conditional $p_2(Y_3|Y_1)$ is the same as $p_1(X_2|X_1)$.  In
particular, its entropy $h_2$ satisfies
$\log |P_2| \geq h_2(Y_1Y_2Y_3) = h_2(Y_1Y_2) + h_2(Y_3|Y_1) =
h_1(X_1X_2)+h_1(X_2|X_1) \geq h_1(X_1,X_2,X_3) = \log |P_1|$ proving
$Q_1 \preceq Q_2$.
\end{ex}

Finally, we give our general proof of Theorem~\ref{th:sufficient}.
To prove the theorem we need three lemmas.  The first lemma is
folklore, and represents the main property of tree decomposition used
for query evaluation.  If $f \in D^X$, $g \in D^Y$ agree on
$X \cap Y$, then $f \Join g$ is the unique tuple $\in D^{X \cup Y}$
that extends both $f$ and $g$.  If
$P_1 \subseteq D^X, P_2 \subseteq D^Y$, then
$P_1 \Join P_2 = \setof{f \Join g}{f \in P_1, g \in P_2}$.

\begin{lmm} \label{lemma:sufficient:1} Let $(T,\chi)$ be a tree
  decomposition for $Q$ and recall that
  $Q \equiv \bigwedge_{t \in\nodes(T)} Q_t$ where $Q_t$ is a
  conjunction of atoms $A$ s.t.  $\vars(A) \subseteq \chi(t)$.  Then,
  for every $\calD$,
  $\hom(Q,\calD) = \Join_{t \in \nodes(t)} \hom(Q_t,\calD)$.
\end{lmm}
\begin{lmm}
  \label{lemma:sufficient:2} Fix a homomorphism
  $\varphi : Q_2 \rightarrow Q_1$, let $(T,\chi)$ be a tree
  decomposition of $Q_2$, $\calD$ be a database instance, and
  $P = \hom(Q_1,\calD)$.  Then, for every node $t \in \nodes(T)$,
  denoting $P'_t \defeq \Pi_{\varphi|_{\chi(t)}}(P)$ we have:
  \begin{align}
    P'_t \subseteq \hom(Q_t, \calD) \label{eq:lemma:sufficient:2}
  \end{align}
\end{lmm}

\begin{proof}
  Every tuple in $\Pi_{\varphi|_{\chi(t)}}(P)$ is the composition
  $f \circ \varphi|_{\chi(t)}$ for some $f \in P$.  The lemma follows
  from the fact that both $\varphi|_{\chi(t)}: Q_t \rightarrow Q_1$
  and $f : Q_1 \rightarrow \calD$ are homomorphisms.
\end{proof}

\begin{lmm} \label{lemma:sufficient:3}
    Let
  $p : P (\subseteq D^{ V}) \rightarrow [0,1]$ be a probability
  distribution, and $h:2^{ V} \rightarrow \R_+$ be its entropy.
    \bi
    \item[(1)]  If
  $\varphi:Y \rightarrow  V$ and $Z \subseteq Y$, then the
  $\varphi|_Z$-pullback of $p$, $\Pi_{\varphi|_Z}(p)$, is equal to the
  $Z$-marginal of $\Pi_\varphi(p)$.  In particular, if
  $h' : 2^Y \rightarrow \R_+$ is the entropy of $\Pi_\varphi(p)$,
  then, $\forall Z \subseteq Y$, $h'(Z) = h(\varphi(Z))$.
  \item[(2)] If
  $\varphi:  V' \rightarrow  V$ and $Y_1, Y_2 \subseteq  V'$,
  then the pull-back distributions $\Pi_{\varphi|_{Y_1}}(p)$ and
  $\Pi_{\varphi|_{Y_2}}(p)$ agree on the common variables
  $Y_1 \cap Y_2$.
  \ei
\end{lmm}

\begin{proof}
  (1) The $\varphi$-pullback $\Pi_\varphi(p)$ is defined to be the
  same as the $\varphi(Y)$-marginal of $p$.  Therefore its
  $Z$-marginal is the $\varphi(Z)$-marginal of $p$.  By definition,
  $\Pi_{\varphi|_Z}(p)$ is also the $\varphi(Z)$-marginal of $p$,
  hence they are equal.  Formally, given $f \in P$:
  \begin{eqnarray*}
    \Pi_\varphi(p)(Z = \Pi_Z(\Pi_\varphi(f))) &=&  \sum_{f': \Pi_Z(\Pi_\varphi(f'))=\Pi_Z(\Pi_\varphi(f))} p(f')\\
  &=& \sum_{f': \Pi_{\varphi(Z)}(f')=\Pi_{\varphi(Z)}(f)} p(f')\\
  &=& \Pi_{\varphi|_Z}(p)(Z = \Pi_{\varphi|_Z}(f))
  \end{eqnarray*}
  because $\Pi_Z \circ \Pi_\varphi = \Pi_{\varphi|_Z}$.  This
  discussion immediately implies that $h'(Z)=h(\varphi(Z))$, forall
  $Z$.

  (2) Let $Z = Y_1 \cap Y_2$.  By claim (1), the $Z$-marginal of
  $\Pi_{\varphi|_{Y_1}}(p)$ is $\Pi_{\varphi|_Z}(p)$ and similarly for
  the $Z$-marginal of $\Pi_{\varphi|_{Y_2}}(p)$, hence they are equal.
\end{proof}

\begin{proof}[Proof of Theorem~\ref{th:sufficient}] Let $\calD$ be any
  database with domain $D$, and let $P = \hom(Q_1,\calD)$.  Consider
  the uniform probability distribution $p : P \rightarrow [0,1]$,
  defined as $p(f) = 1/|P|$ for all tuples $f \in P$, and let
  $h$ be its entropy.  We have $h = \log|P|$ because $p$ is uniform.
  By assumption of the theorem, there exists a homomorphism
  $\varphi: Q_2 \rightarrow Q_1$ and a tree decomposition $(T, \chi)$
  of $Q_2$ such that:
  \begin{align}
  \log|P| = h(\vars(Q_1)) \leq & (E_T \circ \varphi)(h) \label{eq:sufficient:proof:1}
  \end{align}
%
%
  For each $t \in \nodes(T)$, consider the projections of $P$ and $p$
  on $\chi(t)$:
  \begin{eqnarray*}
    P'_t &\defeq& \Pi_{\varphi|_{\chi(t)}}(P),\\
    p'_t &\defeq& \Pi_{\varphi|_{\chi(t)}}(p).
  \end{eqnarray*}
  Lemma~\ref{lemma:sufficient:1} and Lemma~\ref{lemma:sufficient:2}
  imply:
  \begin{eqnarray}
    P' &\defeq&  \Join_{t \in \nodes(T)} P'_t\nonumber\\
    &\subseteq&   \Join_{t \in \nodes(T)} \hom(Q_t,\calD) \nonumber\\
    &=& \ \ \hom(Q_2, \calD) \label{eq:sufficient:proof:2}
  \end{eqnarray}
  We will construct a probability distribution
  $p' : P' \rightarrow [0,1]$, with entropy function
  $h' : 2^{\vars(Q_2)} \rightarrow \R_+$, such that the following
  hold:
  \begin{align}
    h'(\vars(Q_2)) = & E_T(h') \label{eq:sufficient:claim:1} \\
    E_T(h') = & (E_T \circ \varphi)(h) \label{eq:sufficient:claim:2}
  \end{align}
  Assuming the existence of a distribution $p'$ whose entropy function $h'$ satisfies~\eqref{eq:sufficient:claim:1} and~\eqref{eq:sufficient:claim:2},
  the proof of the theorem follows from:
  \begin{align*}
\log |\hom(Q_1,\calD)|&= \log |P|\\
&=  h(\vars(Q_1)) \leq (E_T \circ \varphi)(h) & \mbox{(by Eq.(\ref{eq:sufficient:proof:1}))}\\
         &= E_T(h') & \mbox{(by Eq.(\ref{eq:sufficient:claim:2}))} \\
         &= h'(\vars(Q_2)) & \mbox{(by Eq.(\ref{eq:sufficient:claim:1}))} \\
         &\leq \log |P'| & \mbox{(Since $P'$ is the support of $h'$)} \\
         &\leq \log |\hom(Q_2,\calD)| &\mbox{(By Eq.(\ref{eq:sufficient:proof:2}))}
  \end{align*}
It remains to show how to construct this distribution $p'$
that satisfies~\eqref{eq:sufficient:claim:1} and~\eqref{eq:sufficient:claim:2}.
  We will construct $p'$ by stitching together the pull-back
  distributions $p'_t$, for $t \in \nodes(T)$; this is possible
  because, by Lemma~\ref{lemma:sufficient:3} (2), any two induced
  probabilities $p'_{t_1}, p'_{t_2}$ agree on the common variables
  $\chi(t_1)\cap \chi(t_2)$.

  Formally, we start by listing $\nodes(T)$ in some order,
  $t_1, t_2, \ldots, t_m$, such that each child is listed after its
  parent. Let $P'_i \defeq \Join_{j=1,i} P'_{t_j}$, let $T_i$ be the
  subtree induced by the nodes $\set{t_1, \ldots, t_i}$, and
  $\vars(T_i) = \bigcup_{j=1,i} \chi(t_i)$ its variables.  We
  construct by induction on $i$ a probability distribution
  $p_i' : P_i' \rightarrow [0,1]$ such it agrees with
  $p_{t_1}', \ldots, p_{t_i}'$ on $\chi(t_1), \ldots, \chi(t_i)$
  respectively, and its entropy function
  $h_i' : 2^{\vars(T_i)} \rightarrow \R_+$ satisfies:
  \begin{align}
    h'_i(\vars(T_i)) = & E_{T_i}(h'_i) \label{eq:sufficient:proof:3} \\
    E_{T_i}(h'_i) = & (E_{T_i} \circ \varphi)(h) \label{eq:sufficient:proof:4}
  \end{align}
  To define $p_i'$, we need to extend $p_{i-1}'$ to the variables
  $\vars(T_i) - \vars(T_{i-1}) = \chi(t_i) - \chi(\parent(t_i))$.  We
  define $p_i'$ to satisfy the following: (1) $p_i'$ agrees with
  $p_{t_i}'$ on $\chi(t_i)$, (2) $p_i'$ agrees with $p_{i-1}'$ on the
  $\vars(T_{i-1})$, and (3) $\chi(t_i)$ is independent of
  $\vars(T_{i-1})$ given $\chi(t_i) \cap \chi(\parent(t_i))$.  Notice
  that (1) and (2) are not conflicting because $p_{t_i}'$ agrees with
  any other $p_j'$ on their common variables.  Formally, we define
  $p_i'$ through a sequence of three equations:
  \begin{eqnarray}
    p_i'(\chi(t_i) | \chi(t_i) \cap \chi(\parent(t_i))) &\defeq& p_{t_i}'(\chi(t_i) | \chi(t_i) \cap \chi(\parent(t_i))) \label{eq:sufficient:proof:5}\\
    p_i'(\chi(t_i) | \vars(T_{i-1})) &\defeq& p_i'(\chi(t_i) | \chi(t_i) \cap \chi(\parent(t_i)))  \label{eq:sufficient:proof:6}\\
    p_i'(\vars(T_i)) &\defeq& p_i'(\chi(t_i) | \vars(T_{i-1})) p_{i-1}'(\vars(T_{i-1})) \label{eq:sufficient:proof:7}
  \end{eqnarray}
  We check Eq.(\ref{eq:sufficient:proof:3}):
  \begin{align*}
    h'_i(\vars(T_i))
    &= h_i'(\chi(t_i) | \vars(T_{i-1})) + h_{i-1}'(\vars(T_{i-1})) & \mbox{(by Eq.(\ref{eq:sufficient:proof:7}))} \\
   &= h_i'(\chi(t_i) | \vars(T_{i-1})) + E_{T_{i-1}}(h'_{i-1}) & \mbox{(Induction)}\\
   &= h_i'(\chi(t_i) | \vars(T_{i-1})) + E_{T_{i-1}}(h'_i)  & \mbox{($h'_i$ is identical to $h'_{i-1}$ on $\vars(T_{i-1})$)}\\
   &= h_i'(\chi(t_i) | \chi(t_i)\cap \chi(\parent(t_i))) + E_{T_{i-1}}(h'_i)  & \mbox{(by Eq.(\ref{eq:sufficient:proof:6}))}\\
   &= E_{T_i}(h') &\mbox{(Definition of $E_T$)}
  \end{align*}
  We check Eq.(\ref{eq:sufficient:proof:4}).
  \begin{align*}
      E_{T_i}(h'_i) &= h_i'(\chi(t_i) | \chi(t_i) \cap \chi(\parent(t_i))) + E_{T_{i-1}}(h'_i) & \mbox{(Definition of $E_T$)}\\
  &= h_i'(\chi(t_i) | \chi(t_i) \cap \chi(\parent(t_i))) + (E_{T_{i-1}}\circ \varphi)(h) &  \mbox{(Induction)}\\
  &= h_{t_i}'(\chi(t_i) | \chi(t_i) \cap \chi(\parent(t_i))) + (E_{T_{i-1}}\circ \varphi)(h) &  \mbox{(by Eq.(\ref{eq:sufficient:proof:5}))}\\
  &= h(\varphi(\chi(t_i)) | \varphi(\chi(t_i) \cap \chi(\parent(t_i)))) + (E_{T_{i-1}}\circ \varphi)(h) &  \mbox{(Lemma~\ref{lemma:sufficient:3} (1))}\\
  &= (E_{T_i} \circ \varphi)(h) & \mbox{\quad\quad(Definition of $E_T$)}
  \end{align*}
  This completes the inductive proof.

  By setting $i=m$ (the number of nodes in $T$) in
  Eq.(\ref{eq:sufficient:proof:3}) and (\ref{eq:sufficient:proof:4}),
  we derive Eq.(\ref{eq:sufficient:claim:1}) and
  (\ref{eq:sufficient:claim:2}).

\end{proof}

\subsection{A Necessary Condition}

\label{subsec:necessary}

Next we prove that inequality~\eqref{eqn:another:sufficient}
is also a necessary condition for
containment $Q_1 \preceq Q_2$, when $Q_2$ is acyclic.
Our result answers positively the
conjecture by Kopparty and Rossman~\cite[Sect.3, Discussion 1]{HDE},
in the case when $Q_2$ is acyclic.  To prove the theorem, we consider
some entropy $h$ on which Eq.(\ref{eqn:another:sufficient}) fails, and
prove that the support of its probability distribution, $P$, is a
witness for $Q_1 \not\preceq Q_2$.  The key idea is to use Chan-Yeung's
group-characterizable entropic
functions~\cite{DBLP:conf/isit/Chan07,DBLP:journals/tit/ChanY02}, and show that $P$
can be chosen to be ``totally uniform''.  This allows us to relate
$|\hom(Q_2,\calD)|$ to the right-hand-side of Eq.(\ref{eqn:another:sufficient}).
More precisely, we prove the following.

\begin{thm} \label{th:necessary} Let $Q_2$ be acyclic.
If there exists an entropic function $h$ such that~\eqref{eqn:another:sufficient} does not
hold, namely,
  \begin{align}
      h(\vars(Q_1)) &>  \max_{(T,\chi) \in \td(Q_2)} \max_{\varphi \in \hom(Q_2,Q_1)} (E_T\circ\varphi)(h)
      \label{eqn:necessary}
  \end{align}
  then there exists a database $\calD$ such that
  $|\hom(Q_1,\calD)| > |\hom(Q_2,\calD)|$.
\end{thm}

Together, Theorems~\ref{th:sufficient} and~\ref{th:necessary} prove
that $\bcqca \leq_m \miip$.  To prove Theorem~\ref{th:necessary}, we
need some definitions and lemmas, where we fix a relation
$P \subseteq D^{ V}$, for some set of variables $ V$, let
$p : P \rightarrow [0,1]$ be its uniform distribution
($p(f) \defeq 1/|P|$, for all $f \in P$), and
$h : 2^{ V} \rightarrow \R_+$ its entropy.

\begin{defn}
  We call $P$ {\em totally uniform} if every marginal of $p$ is also uniform.
\end{defn}

For any two sets $X, Y \subseteq  V$, and any tuple
$f_0 \in \Pi_X(P)$, define the $Y$-degree of $f_0$ as
\[\deg_P(Y|X=f_0) \defeq |\setof{\Pi_Y(f)}{f \in P, \Pi_X(f) = f_0}|.\]

\begin{lmm}
  Let $P$ be totally uniform.  Then, for any two sets
  $X,Y \subseteq  V$, the following hold:
  \begin{enumerate}
  \item \label{item:deg:1} $\deg_P(Y|X=f_0)$ is independent of the
    choice of $f_0$, and we denote it by $\deg_P(Y|X)$.
  \item \label{item:deg:2} $\deg_P(Y|X) = |\Pi_{XY}(P)| / |\Pi_X(P)|$
    and $h(Y|X) = \log(\deg_P(Y|X))$.
  \end{enumerate}
\end{lmm}

\begin{proof}
  Item~\ref{item:deg:1} follows from the fact that the $X$-marginal of
  $p$ is uniform and, therefore,
  $p(X=f_0) = \deg(Y|X=f_0) / |\Pi_{XY}(P)|$ is independent of $f_0$.
  For item~\ref{item:deg:2},
  \[|\Pi_{XY}(P)| = \sum_{f_0 \in \Pi_X(P)} \deg_P(Y|X=f_0) =
  |\Pi_X(P)| \cdot \deg_P(Y|X),\]
  and
  \begin{align*}
    h(Y|X) &= h(XY)-h(X)\\ &=\log|\Pi_{XY}(P)|- \log|\Pi_X(P)|
    =\log(\deg_P(Y|X)).
  \end{align*}
\end{proof}

\begin{lmm} \label{lemma:join:uniform} If
  $P_1 \subseteq D^X, P_2 \subseteq D^Y$ and $P_2$ is totally uniform,
  then $|P_1 \Join P_2| \leq |P_1| \cdot \deg_{P_2}(Y|X\cap Y)$.
\end{lmm}

\begin{proof}
  \begin{align*}|P_1 \Join P_2| &\leq \sum_{f\in P_1} \deg_{P_2}(Y|X\cap Y =
  \Pi_{X\cap Y}(f))\\ &= |P_1| \deg_{P_2}(Y|X\cap Y).
  \end{align*}
\end{proof}

\begin{lmm} \label{lemma:group:characterizable} Suppose the $\mii$
  $\max_{i=1,q} E_i(h) \geq 0$ fails for some entropic function $h$.
  Then, for every $\Delta > 0$, there exists a totally uniform
  relation $P$ such that its entropy $h$ satisfies
  $\max_{i=1,q} E_i(h) + \Delta < 0$. In other words, we can find a
  totally uniform witness that fails the inequality with an arbitrary
  large gap $\Delta$.
\end{lmm}
\begin{proof} We use the following result on group-characterizable
    entropic functions~\cite{DBLP:journals/tit/ChanY02}.
  Fix a group $G$.  For every subgroup $G_1 \subseteq G$, denote
  $aG_1 \defeq \setof{ab}{b \in G_1}$.
  An entropic function $h \in \Gamma_n^*$ is called {\em
    group-characterizable} if there exists a group $G$ and subgroups
  $G_1, \ldots, G_n$ such that $h$ is the entropy of the uniform
  probability distribution on
  $P \defeq \setof{(aG_1, \ldots, aG_n)}{a \in G}$.
  Chan and Yeung~\cite{DBLP:journals/tit/ChanY02} proved that
  the set of group-characterizable entropic functions is dense in
  $\Gamma_n^*$; in other words, every $h \in \Gamma_n^*$ is the limit
  of group-characterizable entropic functions.  In particular, if a
  max-linear inequality is valid for all group-characterizable entropic
  functions, then it is also valid for all entropic functions.

  We show that, if $\max_i E_i(h) \geq 0$ fails, then it fails with a
  gap $> \Delta$ on a group-characterizable entropy.  Let $h_0$ be any
  entropic function witnessing the failure:
  $\max_{i=1,q} E_i(h_0) < 0$.  Choose any $\delta > 0$
  s.t. $\max_{i=1,q} E_i(h_0) + \delta < 0$, and define
  $k \defeq \lceil \Delta/\delta \rceil+1$.  Since
  $h \defeq k\cdot h_0 = h_0+h_0+\cdots+h_0$ is also entropic and
  $E_i(k\cdot h_0)=k\cdot E_i(h_0)$ for all $i$, we have that
  $\max_{i=1,q} E_i(h) + k\cdot \delta < 0$, and therefore
  $\max_{i=1,q} E_i(h) + \Delta < 0$.  By Chan-Yeung's density result, we
  can assume that $h$ is group-characterizable.

  Finally, we prove that the set $P$ defining a group-characterizable
  entropy is totally uniform.  This follows immediately from the fact
  that, under the uniform distribution, every tuple
  $(aG_1, \ldots, aG_n) \in P$ has probability
  $|G_1 \cap \cdots \cap G_n|/|G|$, and the marginal probability of
  any tuple $(aG_{i_1}, \ldots, aG_{i_k}) \in \Pi_{i_1\cdots i_k}(P)$
  has probability $|G_{i_1} \cap \cdots \cap G_{i_k}|/|G|$.
  (See Theorem 1 from~\cite{DBLP:conf/isit/Chan07}.)
\end{proof}

\begin{proof}[Proof of Theorem~\ref{th:necessary}]
    Let $(T,\chi)$ be a junction tree (decomposition) of $Q_2$, which exists because acyclic
    queries are chordal. Then,
    \begin{align}
      h(\vars(Q_1)) &>  \max_{(T',\chi) \in \td(Q_2)} \max_{\varphi \in \hom(Q_2,Q_1)}
      (E_{T'}\circ\varphi)(h)  \\
         &\geq \max_{\varphi \in \hom(Q_2,Q_1)} (E_T\circ\varphi)(h).
    \end{align}
    Fix $\Delta$ such that
  $\Delta > \log|\hom(Q_2,Q_1)|$, and let $P \subseteq D^{\vars(Q_1)}$
  be the totally uniform relation given by
  Lemma~\ref{lemma:group:characterizable}, whose entropy $h$
  satisfies:
  \begin{align}
    \log|P| = h(\vars(Q_1)) & > \Delta +  \max_{\varphi \in \hom(Q_2,Q_1)} (E_T\circ\varphi)(h)
  \label{eq:necessary:delta}
  \end{align}
  $P$'s columns are in 1-1 correspondence with
  $\vars(Q_1) = \set{X_1, \ldots, X_n}$.  We annotate each value with
  the column name, thus a tuple $f = (c_1, c_2, \ldots, c_n) \in P$
  becomes \[f = (("X_1",c_1), ("X_2",c_2), \ldots, ("X_n",c_n)).\]
  The
  annotated $P$ is isomorphic with the original $P$, hence still
  totally uniform.  Let $\calD = \Pi_{Q_1}(P)$ be the database
  obtained by projecting the annotated $P$ on the atoms of $Q_1$
  (Eq.(\ref{eq:projection:db})).  We have seen that
  $|\hom(Q_1,\Pi_{Q_1}(P))| \geq |P|$.  We will show that
  $|P| > |\hom(Q_2, \calD)|$, thus $P$ is a witness for
  $Q_1\not\preceq Q_2$.  To do this we need to upper bound
  $|\hom(Q_2,\calD)|$.

  Let $e : \calD \rightarrow Q_1$ be the homomorphism mapping every
  value $("X",c)$ to the variable $X$: this is a
  homomorphism\footnote{For example, let $Q_1 = R(X,X),R(X,Y),S(X,Y)$
    and let $P$ have a single tuple $(a,a)$.  First annotate $P$ to
    $((X,a),(Y,a))$. Then $R^D = \set{((X,a),(X,a)), ((X,a),(Y,a))}$,
    $S^D = \set{((X,a),(Y,a))}$.  Without the annotation, these
    relations would be $R^D = S^D = \set{(a,a)}$, and there is no
    homomorphsims to $Q$, since the tuple in $S^D$ cannot be mapped
    anywhere.}  because, by the definition of $\calD$,
  Eq.(\ref{eq:projection:db}), each tuple
  $f_0 = R_i(("X_{j_1}",c_1), ("X_{j_2}",c_2), \ldots)$ in $\calD$ is
  the projection of some $f \in P$ on the variables $\vars(A)$ of some
  $A \in \atoms(Q_1)$; then $e$ maps $f_0$ to $A$.  If we view a tuple
  $f \in P$ as a function $\vars(Q_1) \rightarrow D$, where $D$ is the
  domain, then $e \circ f$ is the identity function on $\vars(Q_1)$.
  Fix $\varphi \in \hom(Q_2,Q_1)$ and denote:
  \begin{align*}
      \hom_\varphi(Q_2,\calD) & \defeq  \setof{g \in \hom(Q_2,\calD)}{e \circ  g = \varphi}
  \end{align*}
  We have
  \begin{align}
    \hom(Q_2, \calD) &= \bigcup_{\varphi \in \hom(Q_2,Q_1)}  \hom_\varphi(Q_2,\calD) \nonumber\\
    |\hom(Q_2, \calD)| &= \sum_{\varphi \in \hom(Q_2,Q_1)}  |\hom_\varphi(Q_2,\calD)| \label{eq:necessary:1}
  \end{align}
  We will compute an upper bound for $|\hom_\varphi(Q_2,\calD)|$, for
  each homomorphism $\varphi$.  We claim:
  \begin{align}
      \hom_\varphi(Q_2,\calD) & \subseteq \Join_{t \in \nodes(T)} \Pi_{\varphi|_{\chi(t)}}(P) \label{eq:necessary:4}
  \end{align}
  where $\varphi|_{\chi(t)}$ is the restriction of $\varphi$ to
  $\chi(t)$, and $\Pi_{\varphi|_{\chi(t)}}(P)$ is the generalized
  projection (Sec.~\ref{subsec:results:decidable}), i.e. it is a relation with
  attributes $\chi(t)$.  The reason for partitioning $\hom(Q_2,\calD)$
  into subsets $\hom_\varphi(Q_2, \calD)$ is so we can apply
  inequality (\ref{eq:necessary:4}) to each set: notice that the right-hand-side
  depends on $\varphi$.  To prove the claim (\ref{eq:necessary:4}), we
  first observe:
  \begin{align}
      \hom_\varphi(Q_2,\calD) & \subseteq \ \Join_{t \in \nodes(T)} \hom_{\varphi|_{\chi(t)}}(Q_t,\calD) \label{eq:necessary:2}
  \end{align}
  This is a standard property of any join decomposition (not
  necessarily acyclic): every tuple $g \in \hom(Q_2,\calD)$ is the
  join of its fragments $\Pi_{\chi(t)}(g) \in \hom(Q_t,\calD)$, as
  long as the fragments cover all attributes of $g$.  Next we prove
  the following {\em locality property}:
  \begin{align}
      \hom_{\varphi|_{\chi(t)}}(Q_t,\calD) & \subseteq \Pi_{\varphi|_{\chi(t)}}(P) \label{eq:necessary:3}
  \end{align}
  It says that every answer of $Q_t$ on $\calD$ can be found in a
  single row of $P$.  Here we use the fact that $Q_2$ is acyclic
  therefore there exists some $B\in \atoms(Q_2)$ s.t.
  $\vars(B) = \chi(t)$.  Then, any homomorphism
  $g_0 \in \hom_{\varphi|_{\chi(t)}}(Q_t,\calD)$ maps $B$ to some
  tuple $f_0\in \calD$.  By construction of $\calD$, there exists some
  $A \in \atoms(Q_1)$ such that $f_0 \in \Pi_{\vars(A)}(P)$; in
  particular, $f_0 = \Pi_{\vars(A)}(f)$ for some $f \in P$.  Thus
  $g_0$, when viewed as a tuple over variables $\chi(t)$, can be found
  in a single row $f \in P$, more precisely\footnote{We include here
    the rigorous, but rather tedious argument.  Since $g_0$ is a
    homomorphism it ``maps'' the atom $B$ to the tuple $f_0$, meaning
    $(g_0 \circ \vars(B)) = f_0 = (f \circ \vars(A))$ (all are
    functions $[\arity(B)] \rightarrow D$, where $D$ is the domain).
    Since $\vars(B) : [\arity(B)] \rightarrow \chi(t)$ is surjective,
    it has a right inverse, which implies $g_0 = f \circ \psi$ for
    some $\psi$.} $g_0 = \Pi_\psi(f)$, from some function
  $\psi : \chi(t) \rightarrow \vars(Q_1)$.  Noticed that we have used
  in an essential way the fact that $\chi(t)$ is covered by a single
  atom $B$: we will need to remove this restriction later when we prove
  Theorem~\ref{th:decidable} (Lemma~\ref{lemma:necessary:chordal} in
  Section~\ref{sec:decidability}).  From here it is immediate to
  show that $\psi=\varphi|_{\chi(t)}$, by composing with $e$:
  $\varphi|_{\chi(t)} = e \circ g_0 = e \circ f \circ \psi = \psi$
  because $e \circ f$ is the identity on $\vars(Q_1)$.  This completes
  the proof of Eq.(\ref{eq:necessary:3}), which, together with
  Eq.(\ref{eq:necessary:2}), proves the claim
  Eq.(\ref{eq:necessary:4}).

  Finally, we will upper bound the size of the join in
  (\ref{eq:necessary:4}), by applying repeatedly
  Lemma~\ref{lemma:join:uniform}.  This is possible because each
  projection $\Pi_{\varphi|_{\chi(t)}}(P)$ is totally uniform.
  Formally, fix an order of $\nodes(T)$, $t_1, t_2, \ldots, t_m$, such
  that every child occurs after its parent, and compute the join
  (\ref{eq:necessary:4}) inductively, applying
  Lemma~\ref{lemma:join:uniform} to each step.  If
  $S_i \defeq \Join_{j=1,i} \Pi_{\varphi|_{\chi(t_j)}}(P)$, then the
  lemma implies
  $|S_i| = |S_{i-1} \Join \Pi_{\varphi|_{\chi(t_i)}}(P)| \leq
  |S_{i-1}| \deg_{\Pi_{\varphi|_{\chi(t_i)}}(P)}(\chi(t_i) | \chi(t_i)
  \cap \chi(\parent(t_i)))$, and this proves:
  \begin{align}
    |\Join_{t \in \nodes(T)} \Pi_{\varphi|_{\chi(t)}}(P)| \leq
    \prod_{i=1,m}\deg_{\Pi_{\varphi|_{\chi(t_i)}}(P)}(\chi(t_i)|\chi(t_i) \cap \chi(\parent(t_i))
   \label{eq:necessary:5}
  \end{align}
  Let $p' \defeq \Pi_{\varphi|_{\chi(t_i)}}(p)$ be the
  $\varphi|_{\chi(t_i)}$-pullback of $p$.  Its entropy satisfies
  $h'(Z) = h(\varphi(Z)) = (h\circ \varphi)(Z)$ for all
  $Z \subseteq \chi(t_i)$, implying
  $\log\deg_{\Pi_{\varphi|_{\chi(t_i)}}(P)}(Y|Z)= (h \circ \varphi)(Y|Z)$.
  This observation, together with (\ref{eq:necessary:4}) and
  (\ref{eq:necessary:5}) allow us to relate $\hom(Q_2,\calD)$ to
  $(E_T\circ \varphi)(h)$:
  \begin{align*}
    \log |\hom_\varphi(Q_2,\calD)|
    & \leq \sum_{i=1,m} \log\deg_{\Pi_{\varphi|_{\chi(t_i)}}(P)}(\chi(t_i)|\chi(t_i) \cap \chi(\parent(t_i))) \\
    &= \sum_{i=1,m} (h\circ \varphi)((\chi(t_i)|\chi(t_i) \cap \chi(\parent(t_i))) = (E_T\circ\varphi)(h) \\
    &< h(\vars(Q_1)) - \Delta = \log|P| - \Delta &\mbox{(By Eq.(\ref{eq:necessary:delta}))}
  \end{align*}
  Equivalently, $|\hom_\varphi(Q_2,\calD)| < |P|/2^\Delta$.  We sum
  up (\ref{eq:necessary:1}):
  \begin{align*}
      |\hom(Q_2,\calD)|  < |\hom(Q_2,Q_1)|\frac{|P|}{2^\Delta} < |P|
  \end{align*}
  completing the proof.
\end{proof}

We remark that inequality~\eqref{eqn:necessary} is slightly stronger than necessary to prove
containment. In the proof, we only need the inequality to hold for some junction tree.
Conversely, Theorem~\ref{th:sufficient} can also be stated such that we only consider
non-redundant tree decompositions, of which junction trees are a special case.

%% file: sec-completeness.tex
\section{Reducing $\miip$ to $\bcqca$}

\label{sec:max-iti:completeness}

The results of the previous section imply $\bcqca \leq_m \miip$.  We
now prove the converse, $\miip \leq_m \bcqca$; in other words we show
that $\miip$ can be reduced to the containment problem
$Q_1 \preceq Q_2$, with acyclic $Q_2$.

\begin{thm} \label{th:miip:to:bcqca}
  $\miip \leq_m \bcqca$.
\end{thm}

The proof has two parts.  First, we convert the $\miip$ in
Eq.~\eqref{eq:miti:0} into a form that resembles
Eq.\eqref{eqn:another:sufficient}, then we construct
$Q_1$ and $Q_2$.

\subsection{$\miip \leq_m \umiip$}

Consider a general $\miip$ (Eq.(\ref{eq:miti:0})), which we repeat
here:
\begin{align}
    0 \leq \max_{\ell \in [k]} E_\ell(h) \label{eq:miti:e}
\end{align}
where $E_\ell(h) \defeq \sum_{X \subseteq V} c_{\ell,X} h(X)$.
In order to reduce it to a query containment problem, we start by
making the expressions $E_\ell$ uniform.  More precisely, for fixed
natural numbers $n,p,q$, we say that an expression $E$ is
$(n,p,q)$-uniform if:
\begin{align}
  E(h) = n\cdot h(U) + \sum_{j=0,p} h(Y_j|X_j) - q\cdot h(V)
  \label{eq:npq-uniform}
\end{align}
where $V$ is the set of all variables, $U$ is a single variable
called the {\em distinguished variable}, and $X_j, Y_j$, for $j=0,p$, are
(not necessarily distinct) sets of variables, satisfying the following
conditions:
\begin{description}
\item[Chain condition] $X_0 = \emptyset$ and
  $X_j \subseteq Y_{j-1} \cap Y_j$ for $j=1,p$.
\item[Connectedness] $U \in X_j$ for $j=1,p$
\end{description}

A {\em $\umiip$} is a $\miip$, Eq.(\ref{eq:miti:e}), such that
there exist numbers $n,p,q$ and a variable $U$ s.t. all expressions
$E_\ell$ in Eq.(\ref{eq:miti:e}) are $(n,p,q)$-uniform, and have $U$
as a distinguished variable.  Notice that $n,p,q$ and $U$ are the same
in all expressions $E_\ell$.  Clearly, a $\umiip$ is a special
case of a $\miip$.  We prove:

\begin{lmm} \label{lemma:n:n:p} $\miip \leq_m \umiip$.
  Moreover, the reduction can be done in polynomial time.
\end{lmm}

\begin{proof}
  Every $E_\ell$ in Eq.(\ref{eq:miti:e}) has the form
  $\sum_{X \subseteq V} c_{\ell, X} h(X)$.  By expanding each
  positive coefficient as $c_{\ell,X} = 1+1+ \cdots$ and each negative
  coefficient as $c_{\ell,X} = -1 -1 - \cdots$, we can write:
  \begin{align*}
      E_\ell(h) = \sum_{i=1}^{m_\ell} h(Y_i) - \sum_{j=1}^{n_\ell}h(X_j) 
    =  \sum_{i=1}^{m_\ell}h(Y_i) + \sum_{j=1}^{n_\ell}h(V|X_j) - n_\ell\cdot h(V)
  \end{align*}
  Define $X_0 \defeq \emptyset$ and add $h(V|X_0) - h(V)$
  ($=0$) to $E_\ell$:
  \begin{align}
    E_\ell(h) = \sum_{i=1}^{m_\ell}h(Y_i) + \sum_{j=0}^{n_\ell}h(V|X_j) - (n_\ell+1)\cdot h(V)
\label{eq:e:ell}
  \end{align}
  The second sum is a chain, because $X_0=\emptyset$ and every $X_j$
  is contained in $V$.  Let $n \defeq \max_\ell n_\ell$.  We add
  $n-n_\ell$ terms $h(V) - h(V)$ to the expression $E_\ell$,
  resulting in two changes to the expression (\ref{eq:e:ell}): the
  term $-(n_\ell+1) \cdot h(V)$ is replaced by
  $-(n+1)\cdot h(V)$, and the sum $\sum_{i=1,m_\ell}h(Y_i)$
  becomes $\sum_{i=1,m_\ell + n-n_\ell} h(Y_i)$ where the $n-n_\ell$
  new terms are $Y_i \defeq V$.  We combine the two sums
  $\sum_i h(Y_i) + \sum_j h(V| X_j)$ into a single sum by writing
  $h(Y_i)$ as $h(Y_i|\emptyset)$, and thus $E_\ell$ becomes:
  \begin{align}
    E_\ell(h) = \sum_{j=0}^{p_\ell}h(Y_j|X_j) - (n+1)\cdot h(V) \label{eq:e:ell:2}
  \end{align}
  Notice that Eq.(\ref{eq:e:ell:2}) still satisfies the chain condition:
  $X_0=\emptyset$, and
  $X_j \subseteq Y_{j-1} \cap Y_j$ for $j=1,p_\ell$.  Our next step is
  to enforce the connectedness condition.

  Let $U$ be a fresh variable. We will denote by $h$ an entropic
  function over the variables $V$, and by $h'$ an entropic
  function over the variables $UV$.  For $\ell \in [k]$, denote by
  $E_\ell'$ the following expression:
  \begin{align}
\hspace{-3mm}
E_\ell'(h') = (n+1)\cdot h'(U) +  \sum_{j=0}^{p_\ell}h'(UY_j|UX_j) - (n+1)\cdot h'(UV)  \label{eq:e:ell:prime}
  \end{align}
  We claim: $\forall h, 0 \leq \max_\ell E_\ell(h)$ iff
  $\forall h', 0 \leq \max_\ell E_\ell'(h')$.
  For the $\Leftarrow$ direction, assume
  $\forall h': 0 \leq \max_\ell E_\ell'(h')$ and let $h$ be any
  entropic function over the variables $V$.  We extended it to an
  entropic function $h'$ over the variables $UV$, by defining $U$
  to be a constant random variable.  In other words,
  $h'(X) \defeq h(X-\set{U})$ forall $X \subseteq UV$; in
  particular $h'(U)=0$.  Then $E_\ell'(h')=E_\ell(h)$, forall
  $\ell \in [k]$, and the claim follows from
  $0 \leq \max_\ell E_\ell'(h') = \max_\ell E_\ell(h)$.
  For the $\Rightarrow$ direction, let $h'$ be any entropic function
  over the variables $U V$, and denote $h(-) \defeq h'(-|U)$ the
  conditional entropy.  The conditional entropy $h$ is not necessarily
  entropic, but it is the limit of entropic functions (see
  Appendix~\ref{sec:background:it:long}), hence it satisfies
  $0 \leq \max_\ell E_\ell(h)$.  Then,
  $E_\ell'(h') = \sum_{j=0}^{p_\ell}h'(UY_j|UX_j) - (n+1)\cdot h'(UV|U) = \sum_{j=0}^{p_\ell}h(Y_j|X_j) - (n+1) \cdot h(V) =
  E_\ell(h)$,
  and the claim follows from
  $0 \leq \max_\ell E_\ell(h) = \max_\ell E'_\ell(h')$.

  To enforce $X_0=\emptyset$ in the chain condition, we write
  $E_\ell'$ as:
  \begin{align*}
    E_\ell'(h') = n\cdot h'(U) +  \bigl(h'(U) + \sum_{j=0}^{p_\ell}h'(UY_j|UX_j)\bigr) - (n+1)\cdot h'(UV)
  \end{align*}
  Finally, we need to ensure that all numbers $p_\ell$ are equal, and,
  for that, we set $p \defeq 1+\max_\ell p_\ell$ and add $p-p_\ell-1$
  terms $h'(U|U)$ to $E_\ell'(h')$. Comparing it with
  Eq.(\ref{eq:npq-uniform}), the new $E_\ell'$ is an
  $(n,p,n+1)$-uniform expression, proving the lemma.
\end{proof}

\subsection{A Technical Lemma}

The $\umiip$ has some arbitrary $q$, while
Eq.(\ref{eqn:another:sufficient}) has $q=1$.  We prove here a
technical lemma showing that an $(n,p,q)$-uniform $\miip$ is
equivalent to some $\umiip$ with $q=1$.  We do this by
introducing new random variables.

Let $V$ be a set of variables.  For each variable $Z \in V$, we
create $q$ fresh copies $Z^{(\ell)}$, $\ell=1\ldots q$, called adornments of
$Z$.  If $X$ is a set of variables, then $X^{(\ell)}$ is the set where
all variables are adorned with $\ell$.  We will denote by $h$ an
entropic function over the original variables $V$, and by $h'$ an
entropic function over the adorned variables
$V^{(1)}\cdots V^{(q)}$.  If $F = \sum_i c_i h'(X_i^{(\ell_i)})$
is a linear expression over adorned variables, then its erasure,
$\epsilon(F) \defeq \sum_i c_i h(X_i)$, is defined as the expression
obtained by erasing every adornment; we also say that $F$ is an
adornment of $\epsilon(F)$.  Conversely, if $E = \sum_i c_i h(X_i)$ is
an expression over the original variables, then a constant adornment
is an expression of the form
$E^{(\ell)} = \sum_i c_i h'(X_i^{(\ell)})$, i.e. all terms are adorned
by the same $\ell$; clearly $\epsilon(E^{(\ell)}) = E$.

\begin{lmm} \label{lemma:adornment}
  Let $E_1, \ldots, E_k$ be linear expressions over variables $V$,
  and $F_1, \ldots, F_m$ be linear expressions over adorned variables
  $V^{(1)}, \ldots, V^{(q)}$ for some $q \geq 1$, such that
  (a) each $F_j$ is an adornment of some $E_i$,
  i.e. $\epsilon(F_j)=E_i$, and (b) all constant adornments are
  included, i.e for every $E_i$ and every $\ell$ there exists
  $F_j = E_i^{(\ell)}$.  Then the following two statements are
  equivalent:
  \begin{align}
      \forall h: \ q \cdot h(V) \leq & \max_{i \in [k]} E_i(h) \label{eq:equiv:1}\\
    \forall h': \ h'(V^{(1)}\cdots V^{(q)}) \leq & \max_{j=1,m} F_j(h') \label{eq:equiv:2}
  \end{align}
\end{lmm}

\begin{proof}
  (\ref{eq:equiv:1}) $\Rightarrow$ (\ref{eq:equiv:2}) follows from:
  \begin{align*}
    h'(V^{(1)}\cdots V^{(q)}) &\leq \sum_{\ell = 1, q} h'(V^{(\ell)})\\
    & \leq q \max_{\ell=1,q} h'(V^{(\ell)})\\
    &\leq \max_{\ell=1,q} \max_{i \in [k]} E_i^{(\ell)}(h') && \mbox{(Eq.(\ref{eq:equiv:1}) applied to $V^{(\ell)}$)}\\
 &\leq \max_{j=1,m} F_j(h') && \mbox{(Assumption (b))}
  \end{align*}
  (\ref{eq:equiv:2}) $\Rightarrow$ (\ref{eq:equiv:1}) Let $h$ be an
  entropic function over variables $V$.  That means that there
  exists a joint distribution over random variables $V$ whose
  entropy is given by $h$.  For each random variable $Z$, create $q$
  i.i.d. copies $Z^{(\ell)}$, for $\ell=1,q$, and denote by $h'$ the
  entropy function of the new random variables
  $V^{(1)}, \ldots, V^{(q)}$.  Thus, for any adorned set
  $X^{(\ell)}$, $h'(X^{(\ell)}) = h(X)$, and, if
  $E_i = \epsilon(F_j)$, then $E_i(h)=F_j(h')$.  The claim follows
  from:
  \begin{align*}
      q \cdot h(V) &= h'(V^{(1)}) + \cdots + h'(V^{(q)}) && \mbox{(By $h(V) = h'(V^{(\ell)})$, forall $\ell$)} \\
  &= h'(V^{(1)}\cdots V^{(q)}) && \mbox{(Independence)}\\
  &\leq \max_{j=1,m} F_j(h') && \mbox{(Eq.(\ref{eq:equiv:2}))} \\
  &\leq \max_{i \in [k]} E_i(h) && \mbox{(Assumption (a))}
  \end{align*}
\end{proof}

\subsection{$\umiip \leq_m \bcqca$}

\label{subsec:reduction}

Given an $(n,p,q)$-uniform $\miip$ problem (\ref{eq:equiv:1}),
$q \cdot h(V) \leq \max_i E_i$, where
\begin{align}
  E_i = n \cdot h(U) + \sum_{j=0,p} h(Y_{ij} | X_{ij}), \label{eq:e:i}
\end{align}
we will construct two queries $Q_1, Q_2$ such that $Q_1 \preceq Q_2$
iff condition (\ref{eq:equiv:2}) holds, which we have proven is
equivalent to (\ref{eq:equiv:1}).  Recall that the distinguished
variable $U$ occurs everywhere, except in the sets $X_{i0}$ which, by
definition, are $\emptyset$.  We first substitute everywhere the
single variable $U$ with two variables, $U = U_1U_2$.  This does not
affect the $\miip$, since we can simply treat $U_1U_2$ as a joint
variable.

The query $Q_2$ will have one atom for each term of the expression
$E_i$ in (\ref{eq:e:i}), which is possible because, by uniformity, all
expressions $E_i$ have the same number of terms.  In particular, there
will be an atom $R_j$ corresponding to the term $h(Y_{ij} | X_{ij})$,
however, the number of variables $Y_{ij}$ depends on $i$.  For that
reason, we consider their disjoint union, as follows.  For each
variable $V \in V$ and each $i,j$, let $V^{ij}$ be a fresh copy of
$V$; if $W = \set{V_1, V_2, \ldots}$ is a set, then we denote by
$W^{ij} \defeq \set{V_1^{ij}, V_2^{ij}, \ldots}$.  We define
$\tilde Y_j \defeq \bigcup_{i \in [k]} Y_{ij}^{ij}$, for $j=0,p$, and
$\tilde X_j \defeq \bigcup_{i \in [k]} X_{ij}^{i(j-1)}$, for $j=1,p$, and
$\tilde X_0 \defeq \emptyset$.  We notice that
$|\tilde Y_j| = \sum_i |Y_{ij}|$, the sets
$\tilde Y_0, \ldots, \tilde Y_p$ are disjoint, and, since the chain
condition $X_{ij} \subseteq Y_{i(j-1)}$ holds in (\ref{eq:e:i}), we
also have $\tilde X_j \subseteq \tilde Y_{j-1}$; of course,
$\tilde X_j$ is disjoint from $\tilde Y_j$.  We define $Q_2$ as:
\begin{align*}
  Q_2 = S_1(\tilde U_1) \wedge \cdots \wedge S_n(\tilde U_n) \wedge R_0(\tilde X_0 \tilde Y_0 \tilde Z)\wedge \cdots \wedge R_p(\tilde X_p \tilde Y_p \tilde Z)
\end{align*}
All relation symbols are distinct.  The relations $S_1, \ldots, S_n$
are binary, and $\tilde U_1, \ldots, \tilde U_n$ are disjoint sets of
two fresh variables each, and $\tilde Z$ is a fresh set of $k$ variables.
Thus, each relation $R_j$ has arity $(\sum_i (|X_{ij}| + |Y_{ij}|))+k$.
All variables occurring in $R_j$ are distinct (since
$\tilde X_j \subseteq \tilde Y_{j-1}$, which is disjoint from
$\tilde Y_j$) and they occur in the order that corresponds to the order
$X_{1j}\ldots X_{kj} Y_{1j}\ldots Y_{kj}$ of the original variables,
followed by the $k$ variables $\tilde Z$.  Any two consecutive atoms
$R_{j-1}, R_j$ share the variables $\tilde X_j$ and $\tilde Z$, and
therefore the tree decomposition of $Q_2$ consists of $n$ isolated
components plus a chain:
\begin{align}
  T: &&& \set{\tilde U_1} \ \ \ldots \ \ \set{\tilde U_n} \label{eq:hardness:tree}\\ &&&
\set{\tilde X_0, \tilde Y_0,\tilde Z} \stackrel{\tilde X_1,\tilde Z}{-} \set{\tilde X_1, \tilde Y_1,\tilde Z} \stackrel{\tilde X_2, \tilde Z}{-} \set{\tilde X_2 \tilde Y_2, \tilde Z} \cdots \stackrel{\tilde X_p,\tilde Z}{-} \set{\tilde X_p, \tilde Y_p,\tilde Z}
\nonumber
\end{align}
The query $Q_1$ consists of $q$ isomorphic sub-queries:
\begin{align*}
  Q_1 = Q_1^{(1)} \wedge \cdots \wedge Q_1^{(q)}
\end{align*}
which have disjoint sets of variables.  We describe now the subquery
$Q_1^{(\ell)}$.  Its variables consist of adorned copies
$V^{(\ell)}$ of the variables $V$, and the query is in turn a
conjunction of $k$ sub-queries (which are no longer disjoint):
\begin{align*}
  Q_1^{(\ell)} = Q_{1,1}^{(\ell)} \wedge \cdots \wedge Q_{1,k}^{(\ell)}
\end{align*}
To define its atoms, we need some notations.  Recall that the
distinguished variables $U_1U_2$ occur everywhere (except $X_{i0}$
which is empty).  Then, for every $i$, we define the the following
sequences of variables:
\begin{align*}
    \hat X_{ij}^{(\ell)} &=  \underbrace{U_1^{(\ell)}\cdots U_1^{(\ell)}}_{|X_{1j}|} \cdots \underbrace{X_{ij}^{(\ell)}}_{|X_{ij}|}\cdots \underbrace{U_1^{(\ell)}\cdots U_1^{(\ell)}}_{|X_{kj}|}\\
    \hat Y_{ij}^{(\ell)} &=  \underbrace{U_1^{(\ell)}\cdots U_1^{(\ell)}}_{|Y_{1j}|} \cdots \underbrace{Y_{ij}^{(\ell)}}_{|Y_{ij}|}\cdots \underbrace{U_1^{(\ell)}\cdots U_1^{(\ell)}}_{|Y_{kj}|}\\
    \hat Z_i^{(\ell)} &=  \underbrace{U_1^{(\ell)}}_1 \cdots \underbrace{U_1^{(\ell)}}_{i-1} \underbrace{U_2^{(\ell)}}_{i} \underbrace{U_1^{(\ell)}}_{i+1} \cdots \underbrace{U_1^{(\ell)}}_k
\end{align*}
That is, the length of $\hat X_{ij}^{(\ell)}$ is the same as that of
the concatenation $X_{1j}X_{2j}\ldots X_{kj}$, and has the
distinguished variables $U_1^{(\ell)}$ on all positions except the
positions of $X_{ij}$, where it has the adornment $X_{ij}^{(\ell)}$.
(As a special case, $\hat X_{i0}^{(\ell)} = \emptyset$.)
Note that the length of $\hat
X_{ij}^{(\ell)}$ is independent of $i$, and $|\hat X_{ij}^{(\ell)}|  = |\tilde X_j|$
(the variables from $Q_2$).
Similarly for $\hat Y_{ij}^{(\ell)}$.  The sequence $\hat Z_i$ has
length $k$ and contains $U_1^{(\ell)}$ everywhere except for position
$i$ where it has $U_2^{(\ell)}$.  Then, query $Q_{1,i}^{(\ell)}$ is:
\begin{align*}
    Q_{1,i}^{(\ell)} =&  S_1(U^{(\ell)}) \wedge \cdots \wedge S_n(U^{(\ell)}) \wedge \nonumber\\
  & R_0(\hat X_{i0}^{(\ell)} \hat Y_{i0}^{(\ell)} \hat Z_i^{(\ell)}) \wedge R_1(\hat X_{i1}^{(\ell)} \hat Y_{i1}^{(\ell)} \hat Z_i^{(\ell)}) \wedge \cdots \wedge R_p(\hat X_{ip}^{(\ell)} \hat Y_{ip}^{(\ell)} \hat Z_i^{(\ell)})
\end{align*}
Notice that the variables of the atom $R_j$ are just $Y_{ij}^{(\ell)}$
(which contains $U_1^{(\ell)}, U_2^{(\ell)}$, and $X_{ij}^{(\ell)}$), and
some variables are repeated several times.

We start by noticing that every homomorphism $\varphi : Q_2 \rightarrow Q_1$
must map all atoms in the chain $R_0\cdots R_p$ to the same sub-query
$Q_1^{(\ell)}$: this is because the chain is connected and, if one
atom is mapped to an atom whose variables are adorned with $\ell$,
then all atoms must be mapped to atoms adorned similarly with $\ell$.
We claim something stronger, that $\varphi$ maps the entire chain to the
same sub-query $Q_{1,i}^{(\ell)}$.  This is enforced by the variables
$\tilde Z$ of $Q_2$: if one atoms is mapped to the sub-query $Q_{1,i}^{(\ell)}$, then
$\varphi(\tilde Z_i)=U_2^{(\ell)}$ and $\varphi(\tilde Z_{i'})=U_1^{(\ell)}$ forall $i'\neq i$, implying that all
other atoms are mapped to the same sub-query.

By Theorems \ref{th:sufficient} and~\ref{th:necessary}, we have:
\begin{align}
\hspace{-4mm} Q_1\preceq Q_2  && \mbox{ iff } && \forall h', h'(\vars(Q_1)) \leq & \max_{\varphi \in \hom(Q_2,Q_1)} (E_T\circ\varphi)(h')
\label{eq:hardness:equiv:1}
\end{align}
We claim that the following are equivalent:
\begin{align}
  \forall h', h'(\vars(Q_1)) \leq & \max_{\varphi \in \hom(Q_2,Q_1)} (E_T\circ\varphi)(h')
& \mbox{ iff } \nonumber\\
\forall h, q \cdot h( V) \leq & \max_i E_i(h),
\end{align}
where $E_i$ is given by (\ref{eq:e:i}).
The claim implies the theorem: $Q_1 \preceq Q_2$ iff
$\forall h, h( V) \leq \max_i E_i(h)$.  To prove the claim, we will
use Lemma~\ref{lemma:adornment}, and, for that, we need to verify the
conditions of the lemma.  We start by applying the definition of $E_T$
(Eq.(\ref{eq:et})), where $T$ is the tree decomposition of $Q_2$,
Eq.(\ref{eq:hardness:tree}), and obtain (recall that $\tilde X_0 = \emptyset$):
\begin{align*}
    E_T =  h(\tilde U_1) + \cdots + h(\tilde U_n) + h(\tilde Y_0 \tilde Z) + \sum_{j=1,p} h(\tilde X_j \tilde Y_j \tilde Z|\tilde X_j \tilde Z)
\end{align*}
Consider a homomorphism $\varphi \in \hom(Q_2, Q_1)$. By the previous
discussion, it maps all atoms in the chain to the same subquery
$Q_{1,i}^{(\ell)}$ for some $\ell$ and $i$.  We illustrate it by
showing $Q_2$ and $\varphi(Q_2)$ next to each other:
{
\begin{align*}
    Q_2 &=  S_1(\tilde U_1) \wedge \cdots \wedge S_n(\tilde U_n) \wedge R_0(\tilde X_0 \tilde Y_0 \tilde Z)\wedge \cdots \wedge R_p(\tilde X_p \tilde Y_p \tilde Z)\\
    \varphi(Q_2) &=  S_1(U^{(\ell_1)}) \wedge \cdots \wedge S_n(U^{(\ell_n)})
    \wedge
    R_0(\hat X_{i0}^{(\ell)}\hat Y_{i0}^{(\ell)} \hat Z_i^{(\ell)})\wedge \cdots
                 \wedge R_p(\hat X_{ip}^{(\ell)} \hat Y_{ip}^{(\ell)} \hat Z_i^{(\ell)})
\end{align*}
}
Next, we apply the substitution $\varphi$ to $E_T$ to obtain
$E_T \circ \varphi$.  Since each of the original expressions $E_i$ in
Eq.(\ref{eq:e:i}) was $(n,p,q)$-uniform, $U$ occurs in every set
$Y_{ij}$ and $X_{ij}$ (except for $X_{i0}$).
By construction,
$\hat Z_i^{(\ell)}$ is a sequence consisting only of the variables
$U_1^{(\ell)}$ and $U_2^{(\ell)}$, thus the following set inclusions
hold (except for $\hat Z_i^{(\ell)} \subseteq \hat X_{i0}^{(\ell)}$):
$\hat Z_i^{(\ell)} \subseteq \hat X_{ij}^{(\ell)} \subseteq
\hat Y_{ij}^{(\ell)}$, and we obtain:
\begin{align*}
    E_T \circ \varphi &=  h(U^{(\ell_1)}) + \cdots + h(U^{(\ell_n)}) + h(\hat Y_{i0}^{(\ell)} \hat Z_i^{(\ell)})+
\sum_{j=1,p} h(\hat X_{ij}^{(\ell)} \hat Y_{ij}^{(\ell)} \hat Z_i^{(\ell)}|\hat X_{ij}^{(\ell)} \hat Z_i^{(\ell)})\\
                      &=  h(U^{(\ell_1)}) + \cdots + h(U^{(\ell_n)}) + h(Y_{i0}^{(\ell)}) + \sum_{j=1,p} h(Y_{ij}^{(\ell)}|X_{ij}^{(\ell)})
\end{align*}
Clearly its erasure is precisely $\epsilon(E_T \circ \varphi) = E_i$
from Eq.(\ref{eq:e:i}) (recall that $X_{i0}=\emptyset$), proving
condition (a) of the lemma.  Conversely, for each adornment
$E_i^{(\ell)}$ there exists a homomorphism
$\varphi : Q_2 \rightarrow Q_1$ such that
$E_T \circ \varphi = E_i^{(\ell)}$, which proves condition (b),
completing the proof of Th.~\ref{th:miip:to:bcqca}.

\begin{ex}
    We will illustrate the main idea of our reduction from $\miip$ to $\bcqca$ by reducing an $\iip$ to a $\bcqca$.
    Consider the following $\iip$:\footnote{This $\iip$ holds, but our goal is not to check
      it, but to reduce it to $\bcqca$.}
    \begin{align}
        0 &\leq h(X_1) + 2h(X_2) + h(X_3) - h(X_1X_2)-h(X_2X_3) \label{eq:ex:iti:1}
    \end{align}
    We start by rewriting the inequality as:
    \begin{align}
        3h(X_1X_2X_3) &\leq h(X_1)+h(X_2)+h(X_2) + h(X_3) \nonumber \\
      & +h(X_1X_2X_3)+h(X_3|X_1X_2)+h(X_1|X_2X_3)\label{eq:ex:iti:2}
    \end{align}
    From the right-hand-side we derive two queries $Q_1,Q_2$.  Query $Q_1$ has 9
    variables, $X_i^{(\ell)}$, $i=1,3$, $\ell=1,3$, while $Q_2$ has 13
    variables:
    \begin{align*}
      Q_1 &= Q_1^{(1)} \wedge Q_1^{(2)} \wedge Q_1^{(3)} \\
      \ell=1,3:\ \   Q_1^{(\ell)} & = S_1(X_1^{(\ell)})\wedge S_2(X_2^{(\ell)})\wedge S_3(X_2^{(\ell)}) \wedge S_4(X_3^{(\ell)})\\
     & \wedge R_1(X_1^{(\ell)},X_2^{(\ell)},X_3^{(\ell)}) \wedge  R_2(X_1^{(\ell)},X_2^{(\ell)},X_1^{(\ell)},X_2^{(\ell)},X_3^{(\ell)})
     \wedge R_3(X_2^{(\ell)},X_3^{(\ell)},X_1^{(\ell)},X_2^{(\ell)},X_3^{(\ell)}) \\
      Q_2 &= S_1(U_1)\wedge S_2(U_2) \wedge S_3(U_3) \wedge S_4(U_4) \\
          & \wedge R_1(Y_1^0, Y_2^0, Y_3^0) \wedge R_2(Y_1^0, Y_2^0, Y_1^1, Y_2^1, Y_3^1)
          \wedge R_3(Y_2^1, Y_3^1, Y_1^2, Y_2^2, Y_3^2)
    \end{align*}
    We apply Eq.(\ref{eqn:another:sufficient}) to $Q_1$ and $Q_2$.
    $\td(Q_2)$ has a single tree because $Q_2$ is acyclic.  $Q_1$ has 3
    connected components, and $Q_2$ has 5, therefore there are $3^5$
    homomorphisms $Q_2 \rightarrow Q_1$.
    Eq.(\ref{eqn:another:sufficient}) becomes:
    \begin{multline}
  h(X_1^{(1)}X_2^{(1)}X_3^{(1)}X_1^{(2)}X_2^{(2)}X_3^{(2)}X_1^{(3)}X_2^{(3)}X_3^{(3)})\leq\\
   \max_{\ell_1,\ldots,\ell_5 = 1,3}  (h(X_1^{(\ell_1)})+h(X_2^{(\ell_2)})+h(X_2^{(\ell_3)}) + h(X_3^{(\ell_4)})+\\
   h(X_1^{(\ell_5)}X_2^{(\ell_5)}X_3^{(\ell_5)})+h(X_3^{(\ell_5)}|X_1^{(\ell_5)}X_2^{(\ell_5)})+h(X_1^{(\ell_5)}|X_2^{(\ell_5)}X_3^{(\ell_5)}))\label{ex:iip-to-bcqca}
    \end{multline}
  By Theorems~\ref{th:sufficient} and~\ref{th:necessary} and because $Q_2$ is acyclic,
  the $\mii$~\eqref{ex:iip-to-bcqca} holds for all entropic $h$ if and only if
  $Q_1 \preceq Q_2$.
  Moreover Lemma~\ref{lemma:adornment} proves that this $\mii$ is
  equivalent to the $\ii$ in Eq.(\ref{eq:ex:iti:2}), completing the
  reduction from~\eqref{eq:ex:iti:1} to the $\bcqca$ instance $Q_1 \preceq Q_2$.  Our example only illustrated the reduction from $\iip$;
  Lemma~\ref{lemma:n:n:p} addresses the challenges introduced by
  $\miip$.
  \end{ex}

%% file: sec-decidability.tex
\section{Proving Decidability of a Novel Class of $\bcqc$}

In this section, we aim to prove the decidability of our novel class of $\bcqc$ that was
presented  earlier in Section~\ref{subsec:results:decidable}.
In particular, we prove Theorems~\ref{th:decidable} and~\ref{th:product:normal:databases}.
The proofs of both theorems rely on Theorem~\ref{th:simple}, which in turn relies on
Lemma~\ref{lemma:h:domination}.
Therefore, we first prove that lemma in Section~\ref{sec:domination},
and then we prove both theorems in Section~\ref{sec:decidability}.

\subsection{Proof of Lemma~\ref{lemma:h:domination}}

\label{sec:domination}

\begin{lmm}[Re-statement of Lemma~\ref{lemma:h:domination}]
    Let
    $h : 2^{[n]} \rightarrow \R_+$ be any polymatroid.  Then there
    exists a normal polymatroid $h' \in \calN_n$ with the following
    properties:
    \begin{enumerate}
    \item \label{item:domination:1:app} $h'(X) \leq h(X)$, forall $X \subseteq [n]$,
    \item \label{item:domination:2:app} $h'([n]) = h([n])$,
    \item \label{item:domination:3:app} $h'(\set{i}) = h(\set{i})$, forall $i \in [n]$.
    \end{enumerate}
    In addition, there exists a modular function $h''\in \calM_n$ that
    satisfies conditions (\ref{item:domination:1:app}) and
    (\ref{item:domination:2:app}).
  \end{lmm}

  Before we prove the lemma, we need some preliminaries.
Recall that we blurred the distinction between a set of $n$ variables
$ V$ and the set $[n]$.  In this section we will use only $[n]$.  Let
$L \defeq  2^{[n]}$ be the lattice of subsets of $[n]$.  Given a function
$h : L \rightarrow \R_+$, we define its dual $g : L \rightarrow \R_+$
as its M\"obius inverse~\cite{DBLP:conf/pods/KhamisNS16}:
\begin{align}
\forall X: && h(X) = & \sum_{Y: Y \supseteq X} g(Y),
&
g(X) = & \sum_{Y: Y \supseteq X} (-1)^{|Y-X|} h(Y) \label{eq:hg}
\end{align}
For any set $S \subseteq L$ we define:
\begin{align}
  g(S) \defeq & \sum_{X \in S} g(X) \label{eq:gs}
\end{align}
Notice that $g(L) = h(\emptyset)$.
\begin{fact} \label{fact:normal} Let $h : L \rightarrow \R_+$ be any
  function.  Then $h$ is a normal polymatroid (i.e. $h \in \calN_n$) iff its
  M\"obius inverse $g$ satisfies: $g(L) =0$, $g([n]) \geq 0$ and
  $g(X) \leq 0$ forall $X \neq [n]$.
\end{fact}
\begin{proof}
  First we check that the M\"obius inverse of a step function $h_W$
  satisfies the required properties, for $W\subsetneq V$:
  \begin{align*}
    h_W(X) = &
               \begin{cases}
                 0 & \mbox{if $X \subseteq W$} \\
                 1 & \mbox{otherwise}
               \end{cases}
                   &
                     g_W(X) = &
                                \begin{cases}
                                  1 & \mbox{if $X =  V$} \\
                                  -1 & \mbox{if $X = W$} \\
                                  0 & \mbox{otherwise}
                                \end{cases}
  \end{align*}
  The converse follows by observing that every $g$ with the required
  properties is a non-negative linear combination of the $g_W$'s:
  $g = \sum_{W \subsetneq [n]} (-g(W))\cdot g_W$, therefore
  $h = \sum_{W \subsetneq [n]} (-g(W))\cdot h_W$.
\end{proof}
Fact~\ref{fact:normal} can be used, for example, to show that the
parity function $h$ (Example~\ref{ex:parity}) is not normal.  Indeed,
it's M\"obius inverse given by Eq.\eqref{eq:hg} at $\emptyset$ is
$g(\emptyset)=1$, which implies that $h$ is not normal.
Fact.~\ref{fact:normal} will be our key ingredient to prove
Lemma~\ref{lemma:h:domination}: in order to construct the required
normal polymatroid $h'$, we will instead construct its dual $g'$ and
check that it satisfies the conditions in Fact.~\ref{fact:normal}.  We
also need a technical lemma:

\begin{lmm} \label{lemma:max}
  Let $a_1, \ldots, a_n \geq 0$ be $n$ non-negative numbers.  Define:
  \begin{align}
    h(X) = & \max \setof{a_i}{i \in X}  \label{eq:h:max}
  \end{align}
Then $h$ is a normal polymatroid.
\end{lmm}

\begin{proof}
  Assume w.l.o.g. $a_1 \leq a_2 \leq \cdots \leq a_n$ and define
  $\delta_i = a_{i+1} - a_i$ for $i = 0, 1, \ldots, n-1$, where
  $a_0=0$.  Define $g : 2^{[n]} \rightarrow \R$:
  \begin{align*}
    g(X) \defeq &
             \begin{cases}
               a_n & \mbox{ if $X = [n]$} \\
               - \delta_i & \mbox{ if $X=[i]$, ($= \set{1,2,\ldots,i}$), for some $i < n$} \\
               0 & \mbox{ otherwise}
             \end{cases}
  \end{align*}
We check that $g$ is the dual of $h$ by verifying:
\begin{align*}
  h(X) = & a_{\max(X)} = - \delta_{\max(X)} - \delta_{\max(X)+1} - \cdots -\delta_{n-1} + a_n = \sum_{Y: X \subseteq Y} g(Y)
\end{align*}
We assumed above that $\max(\emptyset) =0$.
\end{proof}

Finally, we need to recall the definitions of the {\em conditional
  entropy} and the {\em conditional mutual information}:
\begin{align}
  h(i|X) = & h(\set{i}\cup X) - h(X)\nonumber\\ I(i;j|X) = & h(\set{i}\cup X) +  h(\set{j}\cup Y) - h(X) - h(\set{i,j} \cup X)
\label{eq:cond:i}
\end{align}
and observe that, denoting
$[X,Y] \defeq \setof{Z}{X \subseteq Z \subseteq Y}$, we have:
\begin{align}
  h(X)         = & g([X,[n]]) \label{eq:g:to:h} \\
  h(i|X)       = & -g([X,[n]-\set{i}]) \label{eq:g:to:cond} \\
  I(i;j|X)     = & -g([X,[n]-\set{i,j}]) \label{eq:g:to:i}
\end{align}

We are now ready to prove Lemma~\ref{lemma:h:domination}.
\bp[Proof of Lemma~\ref{lemma:h:domination}]
We will proceed by
induction on $n$.  Split the lattice $L = 2^{[n]}$ into two disjoint
sets $L = L_1 \cup L_2$ where:
\begin{align*}
L_1 = & [\emptyset, [n-1]] & L_2 = [\set{n}, [n]]
\end{align*}
In other words, $L_1$ contains all subsets without $n$, while $L_2$
contains all subsets that include $n$.  Then:

\begin{itemize}
\item $g(L_2) = h(\set{n})$.  It follows $g(L_1) = - h(\set{n})$.
\item Subtract $h(\set{n})$ from $g([n])$ and add it to $g([n-1])$,
  and call $g_1, g_2$ the new functions on $L_1, L_2$ respectively.
  Formally:
  \begin{align*}
    g_1(X) = &
               \begin{cases}
                 g([n-1]) + h(\set{n}) & \mbox{ if $X=[n-1]$}\\
                 g(X) & \mbox{ if $X \subset [n-1]$}
               \end{cases}
\\
    g_2(X\cup \set{n}) = &
               \begin{cases}
                 g([n]) - h(\set{n}) & \mbox{ if $X=[n-1]$}\\
                 g(X \cup \set{n}) & \mbox{ if $X \subset [n-1]$}
               \end{cases}
  \end{align*}
  Notice that $g_1(L_1)=0$ and $g_2(L_2)=0$.
\item One can check that the dual\footnote{Strictly speaking we cannot
    talk about the dual of $g_2$ because we defined the dual only for
    functions $g : 2^{[m]} \rightarrow \R$.  However, with some abuse,
    we identify the lattice $L_2$ with $2^{[n-1]}$, and in that sense
    the dual of $g_2 : L_2 \rightarrow \R$ is a function
    $h_2 : L_2 \rightarrow \R$.} of $g_2$ is the {\em conditional}
  polymatroid\footnote{Proof:
    $h_2(X) = \sum_{Y: X \subseteq Y \subseteq [n]} g_2(Y) = \sum_{Y:
      X \subseteq Y \subseteq [n]} g(Y) - h(\set{n}) = h(X) -
    h(\set{n}) = h(X|\set{n})$.}, defined as $h_2:L_2 \rightarrow \R$:
  \begin{align*}
    \forall X \in L_2:     h_2(X) \defeq & h(X|\set{n})
  \end{align*}
\item We apply induction to $h_2$ and obtain a normal polymatroid
  $h_2' : L_2 \rightarrow \R$ satisfying properties (\ref{item:domination:1}),
  (\ref{item:domination:2}), and (\ref{item:domination:3}) that are stated in Lemma~\ref{lemma:h:domination}:
  \begin{align*}
    h_2'(X) \leq h_2(X) = & h(X|\set{n}) \\
    h_2'([n]) = h_2([n]) = & h([n]|\set{n}) \\
    h_2'(\set{i,n}) = h_2(\set{i,n}) = & h(\set{i}|\set{n}) & \mbox{ since $\set{i,n}$ is an atom in $L_2$}
  \end{align*}
  Notice that $h_2'(\set{n})= 0$, since $\set{n}$ is the bottom of
  $L_2$.  Let $g_2'$ be the dual of $h_2'$, thus $g_2'(X) \leq 0$
  forall $X \neq [n]$ (because $h_2'$ is normal).
\item One can check that the dual of $g_1$ is the
  function\footnote{Proof:
    \begin{align*}
      h_1(X) = & \sum_{Y: X \subseteq Y \subseteq [n-1]} g_1(Y)
                 = h(\set{n}) + \sum_{Y: X \subseteq Y \subseteq [n-1]} g(Y)\\
      =& h(\set{n}) + \sum_{Y: X \subseteq Y \subseteq [n]} g(Y) - \sum_{Y: X \subseteq Y \subseteq [n-1]} g(Y\cup\set{n})\\
      = & h(\set{n}) + h(X) - h(X \cup \set{n}) = I(X;\set{n})
    \end{align*}
  }
  \begin{align*}
     h_1(X) \defeq & I(X;\set{n})
  \end{align*}
  This is no longer a polymatroid.  Instead, here we use
  Lemma~\ref{lemma:max} and define the normal polymatroid
  $h_1' : L_1 \rightarrow \R$:
  \begin{align*}
  h_1'(X) \defeq & \max_{i\in X} h_1(\set{i}) = \max_{i \in X} I(\set{i};\set{n})
  \end{align*}
  Let $g_1': L_1 \rightarrow \R$ be its dual.  Thus, $g_1'(X) \leq 0$
  forall $X \neq [n-1]$, and
  $g_1'([n-1]) = \max_{i \in [n-1]} I(\set{i}; \set{n})$.
\item We combine $g_1', g_2'$ into a single function
  $g' : L (=L_1 \cup L_2) \rightarrow \R$ as follows.  $g'$ agrees
  with $g_1'$ on $L_1$ and with $g_2'$ on $L_2$ except that we
  subtract a mass of $h(\set{n})$ from $g_1'([n-1])$ and add it to
  $g_2'([n])$.  Formally:
  \begin{align*}
    g'(X) \defeq &
                   \begin{cases}
                     g_2'([n]) + h(\set{n}) & \mbox{if $X=[n]$}\\
                     g_1'([n-1]) - h(\set{n}) & \mbox{if $X = [n-1]$}  \\
                     g_1'(X) & \mbox{if $X \in L_1, X \neq [n-1]$}\\
                     g_2'(X) & \mbox{if $X \in L_2, X \neq [n]$}
                   \end{cases}
  \end{align*}
\item We claim that for every $X \neq [n]$, $g'(X) \leq 0$.  This is
  obvious for all cases above (since $g_1', g_2'$ are normal), except
  when $X = [n-1]$.  Here we check:
  $g'([n-1]) = g'_1([n-1]) - h(\set{n}) = \max_{i \in [n-1]}
  I(\set{i};\set{n}) - h(\set{n}) \leq 0$ because
  $I(\set{i};\set{n}) \leq h(\set{n})$.
\item Denote $h': L (=L_1 \cup L_2) \rightarrow \R$ the dual of $g'$;
  we have established that $h'$ is a normal polymatroid.  The
  following hold:
  \begin{align}
\forall x \in L_1:
    h'(X) = & \sum_{Y: X \subseteq Y \subseteq [n]} g'(Y) \nonumber\\= &
 \sum_{Y: X \subseteq Y \subseteq [n-1]} g'(Y) +  \sum_{Y: X \subseteq Y \subseteq [n-1]} g'(Y\cup\set{n})\nonumber\\
 = & \sum_{Y: X \subseteq Y \subseteq [n-1]} g_1'(Y) +  \sum_{Y: X \subseteq Y \subseteq [n-1]} g_2'(Y\cup\set{n})
 \nonumber\\= &h_1'(X) + h_2'(X \cup \set{n}) \label{eq:aux1}\\
\forall X \in L_2:
    h'(X) = & \sum_{Y: X \subseteq Y \subseteq [n]} g'(Y) \nonumber\\=& h(\set{n})  + \sum_{Y: X \subseteq Y \subseteq [n]} g_2'(Y) = h(\set{n}) + h_2'(X) \label{eq:aux2}
  \end{align}
\item We check that $h'$ satisfies properties
  (\ref{item:domination:1}), (\ref{item:domination:2}), and
  (\ref{item:domination:3}) that are stated in Lemma~\ref{lemma:h:domination}:
  \begin{align*}
\forall X \in L_1:
 h'(X) = & h_1'(X) + h_2'(X \cup \set{n})  &\mbox{ by Eq.(\ref{eq:aux1})} \\
    \leq & h_1(X) + h_2(X \cup \set{n}) \\= & I(X;\set{n}) + h(X | \set{n}) = h(X)\\
\forall X \in L_2:
 h'(X) = & h(\set{n}) + h_2'(X)  &\mbox{ by Eq.(\ref{eq:aux2})}  \\
    \leq & h(\set{n}) +  h_2(X) \\=& h(\set{n}) + h(X | \set{n}) = h(X)\\
 h'([n]) = &  h(\set{n}) + h_2'([n])  &\mbox{ by Eq.(\ref{eq:aux2})}  \\
         = & h(\set{n}) + h_2([n]) \\=& h(\set{n}) + h([n]|\set{n}) = h([n]) \\
\forall i \in [n-1]:
  h'(\set{i}) = & h_1'(\set{i}) + h_2'(\set{i,n}) &\mbox{ by Eq.(\ref{eq:aux1})}\\
             =  & h_1(\set{i}) + h_2(\set{i,n}) \\=& I(\set{i};\set{n}) + h(\set{i}|\set{n}) = h(\set{i}) \\
    h'(\set{n}) = & h(\set{n}) + h_2'(\set{n}) = h(\set{n}) + 0 &\mbox{ by Eq.(\ref{eq:aux2})}
  \end{align*}
\end{itemize}

This completes the proof.
\ep

We illustrate the main idea of the above proof using the following example,
which is based on the parity
function, also shown in Fig.~\ref{fig:example}.

\input{fig-example}

\begin{ex}
  \label{ex:parity:g}
  Recall the parity function, and it's M\"obius inverse:
  \begin{eqnarray*}
    h(\emptyset)=0,\quad  h(1)=h(2)=h(3)=1, \\ h(12)=h(13)=h(23)=h(123)=2,\\
    g(123)=2,\quad g(12)=g(13)=g(23)=0,\\ g(1)=g(2)=g(3)=-1,\quad g(\emptyset)=+1.
  \end{eqnarray*}
  The parity function is not normal, because $g(\emptyset) > 0$.  The
  lattice $L=2^{[3]}$ is shown on the top left of
  Fig.~\ref{fig:example}.

  We partition $L = L_1 \cup L_2$, and move a mass of $+1$ from
  $g(123)$ to $g(12)$ (so that both lattices are balanced,
  i.e. $g_1(L_1)=0, g_2(L_2)=0$); this is show in the top right. We
  compute $h_1, h_2$ from $g_1, g_2$.  Notice that $h_1$ is not a
  polymatroid.

  We define $h_1'$ using the max-construction (Lemma~\ref{lemma:max})
  and define $h_2' = h_2$ (since it is already normal).  Notice that
  $h_1'=0$.  From $h_1', h_2'$ we compute $g_1', g_2'$.  Lower right
  of Figure~\ref{fig:example}.

  Finally we combine the two functions $g_1', g_2'$ and obtain the
  functions $h', g'$ shown in the lower left.  $h'$ is normal, is
  dominated by $h$, and agrees with $h$ on the atoms and the maximum
  element of the lattice.
\end{ex}

\subsection{Proof of Theorem~\ref{th:decidable} and~\ref{th:product:normal:databases}}

\label{sec:decidability}

\begin{thm}[Re-statement of Theorem~\ref{th:decidable}] Checking $Q_1 \preceq Q_2$ is
decidable in exponential time when $Q_2$ is chordal and admits a simple
junction tree.
\end{thm}

\begin{thm}[Re-statement of Theorem~\ref{th:product:normal:databases}] Let $Q_2$ be chordal,
\begin{itemize}
    \item[(i)] If $Q_2$ admits a totally disconnected junction tree,
        then $Q_1 \not\preceq Q_2$ if and only if there is a product witness.
    \item[(ii)] If $Q_2$ admits a simple junction tree, then $Q_1 \not\preceq Q_2$
    if and only if there exists a normal witness.
\end{itemize}
\end{thm}

In order to prove the above theorems, we need a technical lemma.
In Theorem~\ref{th:necessary} we proved that, when $Q_2$ is acyclic and
Eq.(\ref{eqn:another:sufficient}) fails, then $Q_1 \not\preceq Q_2$.
Our next lemma is a variation of that result: when $Q_2$ is chordal and
Eq.(\ref{eqn:another:sufficient}) fails on a normal entropic function,
then $Q_1 \not\preceq Q_2$. Recall that a junction tree is a special tree decomposition.

\begin{lmm} \label{lemma:necessary:chordal} Let $Q_2$ be chordal and admit a simple junction tree $T$, and let $E_T$ be its linear
  expression, Eq.(\ref{eq:et}).  If there exists a normal entropic
  function $h$ (i.e. with a non-negative I-measure) such that:
  \begin{align}
    h(\vars(Q_1)) & >  \max_{\varphi \in \hom(Q_2,Q_1)} (E_T\circ\varphi)(h)
\label{eq:necessary:chordal}
  \end{align}
  then there exists a database instance $\calD$ such that
  $|\hom(Q_1,\calD)| > |\hom(Q_2,\calD)|$.
\end{lmm}

We first show how to use the lemma and the essentially-Shannon
inequalities in Theorem~\ref{th:simple} to prove
Theorems~\ref{th:decidable} and~\ref{th:product:normal:databases}.
Assume $Q_2$ is chordal and has a simple junction tree $T$.  We prove: $Q_1 \preceq Q_2$ iff
Eq.(\ref{eqn:another:sufficient}) holds.  It suffices to prove that
Eq.(\ref{eqn:another:sufficient}) is necessary, because sufficiency
follows from Theorem~\ref{th:sufficient}.  Suppose
Eq.(\ref{eqn:another:sufficient}) fails.  Then there exists an entropic function $h$ such that (\ref{eq:necessary:chordal}) holds where $T$ in (\ref{eq:necessary:chordal}) is a simple junction tree of $Q_2$.
Since $T$ is simple,
the conditional linear expressions on the right-hand-side of \eqref{eq:necessary:chordal} are also simple.
By Theorem~\ref{th:simple}, there exists a {\em normal} entropic function $h$ such that (\ref{eq:necessary:chordal}) holds.
Then, by Lemma~\ref{lemma:necessary:chordal},
$Q_1 \not\preceq Q_2$.  This proves that
Eq.(\ref{eqn:another:sufficient}) is necessary and sufficient for
containment.  Furthermore, Eq.(\ref{eqn:another:sufficient}) is
decidable, since it is an essentially-Shannon inequality, and this
completes the proof of Theorems~\ref{th:decidable}.  The proof of
Theorem~\ref{th:product:normal:databases} follows immediately from the
fact that the set of normal entropic functions $\calN_n$ is the cone
generated by the entropies of normal relations, and the set of modular
functions $\calM_n$ is the cone generated by the entropies of product
relations.

It remains to prove Lemma~\ref{lemma:necessary:chordal}; the lemma
generalizes Theorem 3.2 of~\cite{HDE} to arbitrary vocabularies
(beyond graphs).  To prove the theorem, we will update the proof of
Theorem~\ref{th:necessary}, where we used acyclicity of $Q_2$: more
precisely we need to re-prove the locality property,
Eq.(\ref{eq:necessary:3}).  We repeat it here:
\begin{align*}
  \hom_{\varphi|_{\chi(t)}}(Q_t,\calD) \subseteq & \Pi_{\varphi|_{\chi(t)}}(P)
\end{align*}
We start by observing that this property fails in general.
\begin{ex} Let $Q_1 = R(X_1,X_2),S(X_2,X_3),T(X_3,X_1)$ and
  $Q_2 = R(Y_1,Y_2),S(Y_2,Y_3),T(Y_3,Y_1)$ (they are identical).
  Consider the parity function in Example~\ref{ex:parity}; more
  precisely, this is the entropy of the relation
  $P = \setof{(X_1,X_2,X_3)}{X_1,X_2,X_3 \in \set{0,1}, X_1\oplus
    X_2\oplus X_3 = 0}$, which we show here for clarity:
  \begin{center}
    \begin{tabular}[t]{l|c|c|c|} \cline{2-4}
            & 0 & 0 & 0 \\
      $P=$  & 0 & 1 & 1 \\
            & 1 & 0 & 1 \\
            & 1 & 1 & 0 \\ \cline{2-4}
    \end{tabular}
  \end{center}
  Recall that the entropy of $P$ is not a normal entropic function
  (Sec.~\ref{sec:domination}).  This relation is perfectly uniform (in
  fact it is a group characterization).  Computing
  $\calD = \Pi_{Q_1}(P)$ we obtain
  $R^D = S^D = T^D = \set{(0,0),(0,1),(1,0),(1,1)}$.  $Q_2$ is a
  clique, with a bag $Q_t = Q_2$, and $\hom(Q_t,\calD)$ contains one
  extra triangle, $(1,1,1)$, which is in no single row of $P$.
\end{ex}

The example shows that we need to use in a critical way the fact that
the counterexample $h$ is a normal entropic function, $h \in \calN_n$.
To use this fact, we will describe a class of relations whose entropic
functions generate precisely the cone $\calN_n$, and prove that these
are precisely the normal relations (Def.~\ref{def:normal:relation}).

Before we start, we review a
basic concept, which we call ``domain-product'', first introduced by
Fagin~\cite{DBLP:journals/jacm/Fagin82} to prove the existence of an
Armstrong relation for constraints defined by Horn clauses, and later
used by Geiger and Pearl~\cite{GeigerPearl1993} to prove that
Conditional Independence constraints on probability distributions also
admit an Armstrong relation.  The same construction appears under the
name ``fibered product'' in~\cite{HDE}.

\begin{defn} \label{def:domain:product} Fix two domains $D_1, D_2$.
  For any two tuples $f \in D_1^{ V}$, $g \in D_2^{ V}$ we
  define $f \otimes g \in (D_1 \times D_2)^{ V}$ as the function
  $(f \otimes g)(x) \defeq (f(x),g(x))$ for all $x \in  V$.  The
  {\em domain product} of two relations $P_1 \subseteq D_1^{ V}$,
  $P_2 \subseteq D_2^{ V}$ is
  $P_1 \otimes P_2 \defeq \setof{f \otimes g}{f \in P_1, g \in P_2}$.
  If $p_1, p_2$ are probability distributions on $P_1, P_2$
  respectively, then their {\em product} $p_1\cdot p_2$ is the
  probability distribution
  $(p_1 \cdot p_2)(f,g) \defeq p_1(f) \cdot p_2(g)$ on
  $P_1 \otimes P_2$.
\end{defn}

The following basic fact relates to the above definition: if $h_1$ and $h_2$ are two entropic functions,
then $h_1 + h_2$ is also entropic. In particular, if $h_i$ is the entropy of
$p_i : P_i \rightarrow [0,1]$, then $h_1 + h_2$ is the entropy of
$p_1 \cdot p_2 : P_1 \otimes P_2 \rightarrow [0,1]$, where
$P_1 \otimes P_2$ is the domain product.

Now we are ready to prove Lemma~\ref{lemma:necessary:chordal}.
Consider the normal entropic function $h$ given by
Lemma~\ref{lemma:necessary:chordal}.  We can assume w.l.o.g. that $h$
is a sum of step functions\footnote{Suppose the contrary, that the
  inequality holds for all functions $h$ that are sums of step
  functions.  Then it holds for all linear combinations
  $\sum_W c_W h_W$ where $c_W \geq 0$ are integer coefficients.  If an
  inequality holds for $h$, then it also holds for $\lambda \cdot h$
  for any constant $\lambda > 0$; it follows that the inequality holds
  for all linear combinations $\sum_W c_W h_W$ where $c_W \geq 0$ are
  rationals.  The topological closure of these expressions is
  $\calN_n$, contradicting the fact that the inequality fails on some
  $h \in \calN_n$.}, $h = \sum_i h_{W_i}$, where each $h_{W_i}$ is a
step function (not necessarily distinct).  Recall from Section~\ref{sec:background:it:short} that $P_{W_i}$ is
the 2-tuple relation whose entropy is $h_{W_i}$; to reduce clutter we
denote here $P_{W_i}$ by $P_i$.  Then $h$ is the entropy of their
domain-product (Def~\ref{def:domain:product}),
$P = P_1 \otimes P_2 \otimes \cdots \otimes P_m$. One can check that
$P$ is totally uniform (it is even a group realization).  We now prove
the locality property, Eq.(\ref{eq:necessary:3}), using the fact that
$P$ is a domain product, which allows us to rewrite
Eq.(\ref{eq:necessary:3}) as:
\begin{align*}
  \hom_{\varphi|_{\chi(t)}}(Q_t,\calD_1 \otimes \cdots \otimes \calD_m) \subseteq
& \Pi_{\varphi|_{\chi(t)}}(P_1 \otimes \cdots \otimes P_m)
\end{align*}
It suffices prove that
$\hom_{\varphi|_{\chi(t)}}(Q_t,\calD_i) \subseteq
\Pi_{\varphi|_{\chi(t)}}(P_i)$ for each $i$.  Recall that $P_i$ has
two tuples, $P_i = \set{f_1, f_2}$, where $f_1=(1,1,\ldots,1)$ and
$f_2$ has values $1$ on positions $\in W$ and values $2$ on positions
$\not\in W$, for some set of attributes $W$.  Fix a tuple
$g \in \hom_{\varphi|_{\chi(t)}}(Q_t,\calD_i)$; we must prove that
either $g \in \Pi_{\varphi|_{\chi(t)}}(f_1)$ or
$g \in \Pi_{\varphi|_{\chi(t)}}(f_2)$.  If $g$ maps every variable in
$\vars(Q_t)$ to $1$, then the first condition holds, so assume that
$g$ maps some variable $Y \in \vars(Q_t)$ to $2$; in particular,
$\varphi(Y) \not\in W$.  We must prove that, for every variable $Y'$,
if $\varphi(Y') \not\in W$ then $g(Y')=2$.  Here we use the fact that
$Q_2$ is chordal, hence $Q_t$ is a clique, thanks to Fact~\ref{fct:BCQC:proj}. Therefore, there exists
$B \in \atoms(Q_t)$ that contains both $Y$ and $Y'$.
Since $g$ is a
homomorphism, it maps $B$ to some tuple in
$\Pi_{\varphi(\vars(B))}(P)$; since both
$\varphi(Y), \varphi(Y') \not\in W$, this tuple must have the value 2
on both positions (they can be identical: $\varphi(Y) = \varphi(Y')$).
It follows that all variables $Y'$ s.t. $\varphi(Y') \not\in W$ are
mapped to $2$, proving that $g \in \Pi_{\varphi|_{\chi(t)}}(f_2)$.
This proves the local property, Eq.(\ref{eq:necessary:3}).  The rest
of the proof of Theorem~\ref{th:necessary} remains unchanged, and this
completes the proof of Lemma~\ref{lemma:necessary:chordal}.

%% file: fig-example.tex
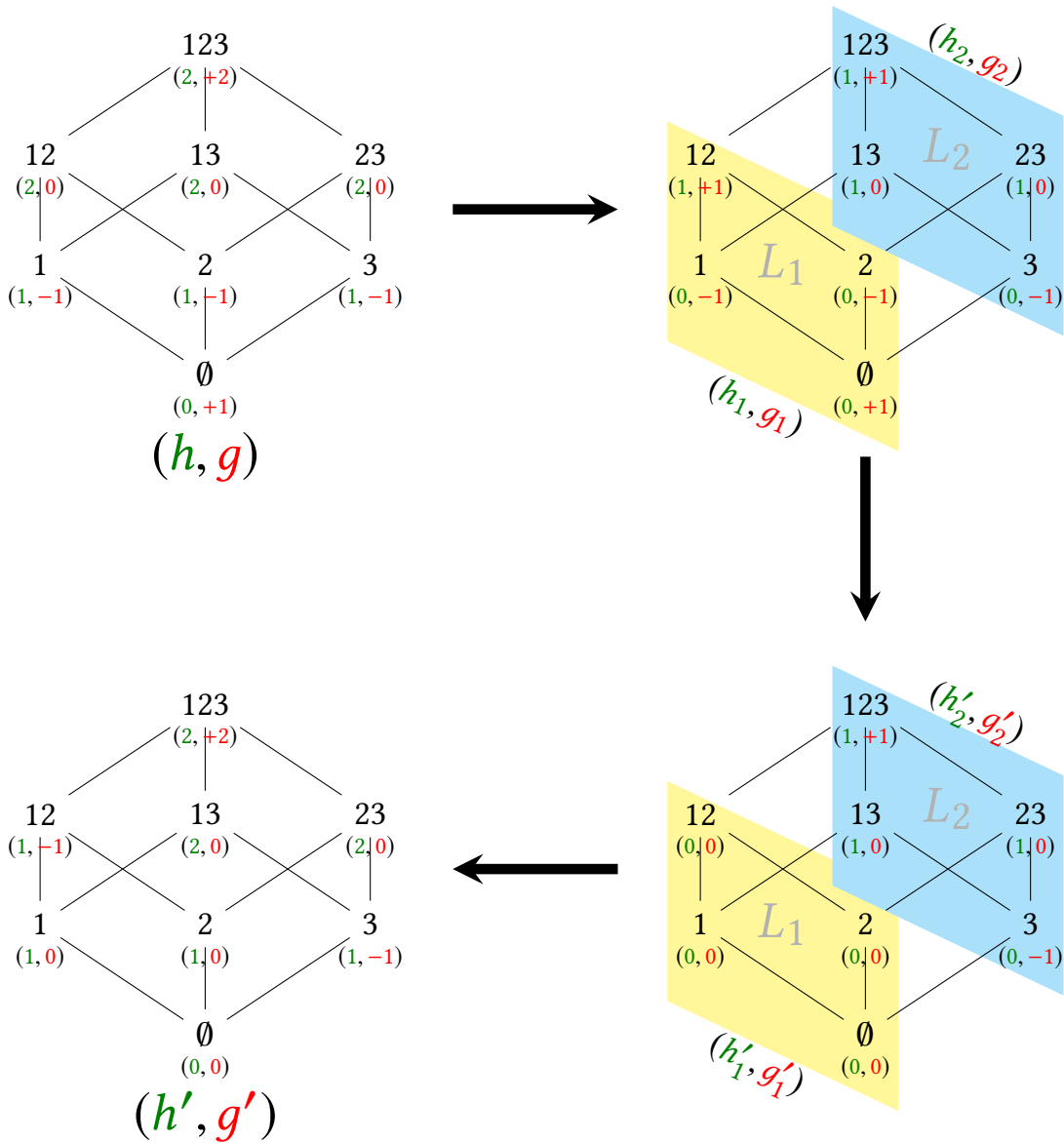
\begin{figure*}[th!]
\centering

\tikzstyle{lattice-node} = []
\tikzstyle{lattice-edge} = []
\newcommand{\hgpair}[2]{$({\color{DarkGreen}#1}, {\color{red}#2})$}
\newcommand{\hglabel}[3]{
    \node[scale = .7] at ($(#1)+(0, -.3)$){\hgpair{#2}{#3}};
}
\newcommand{\sublattice}[1]{
    \begin{scope}[shift={(-2,1)}]
        \fill[#1] (.2,-.2)--(2.3,-1.2)--(2.3, -3.2)--(.2,-2.2)--cycle;
    \end{scope}
}

\begin{tikzpicture}[scale=1.3, every node/.style={transform shape}]


    \begin{scope}[shift={(2,-2)}]
        \node[lattice-node] at (0, -1.5) (0) {$\emptyset$};
        \node[lattice-node] at (-1.5, -.5) (x) {$1$};
        \node[lattice-node] at (0, -.5) (y) {$2$};
        \node[lattice-node] at (1.5, -.5) (z) {$3$};
        \node[lattice-node] at (-1.5, .5) (xy) {$12$};
        \node[lattice-node] at (0, .5) (xz) {$13$};
        \node[lattice-node] at (1.5, .5) (yz) {$23$};
        \node[lattice-node] at (0, 1.5) (xyz) {$123$};

        \path[lattice-edge] (0) edge (x);
        \path[lattice-edge] (0) edge (y);
        \path[lattice-edge] (0) edge (z);
        \path[lattice-edge] (x) edge (xy);
        \path[lattice-edge] (x) edge (xz);
        \path[lattice-edge] (y) edge (xy);
        \path[lattice-edge] (y) edge (yz);
        \path[lattice-edge] (z) edge (xz);
        \path[lattice-edge] (z) edge (yz);
        \path[lattice-edge] (xy) edge (xyz);
        \path[lattice-edge] (xz) edge (xyz);
        \path[lattice-edge] (yz) edge (xyz);

        \hglabel{0}  {0}{+1}
        \hglabel{x}  {1}{-1}
        \hglabel{y}  {1}{-1}
        \hglabel{z}  {1}{-1}
        \hglabel{xy} {2}{0}
        \hglabel{xz} {2}{0}
        \hglabel{yz} {2}{0}
        \hglabel{xyz}{2}{+2}

        \node[scale=1.5] at (0, -2.25) {\hgpair{h}{g}};
    \end{scope}

    \begin{scope}[shift={(8,-2)}]
        \sublattice{yellow!50!white}
        \node[gray!60!white, scale = 1.5] at (-.75, -.5) {$L_1$};
        \begin{scope}[shift={(1.5,1.05)}]
            \sublattice{cyan!30!white}
            \node[gray!60!white, scale = 1.5] at (-.75, -.5) {$L_2$};
        \end{scope}
        \node[rotate = -25] at (-1, -1.8) {\hgpair{h_1}{g_1}};
        \node[rotate = -25] at (1, 1.4) {\hgpair{h_2}{g_2}};

        \node[lattice-node] at (0, -1.5) (0) {$\emptyset$};
        \node[lattice-node] at (-1.5, -.5) (x) {$1$};
        \node[lattice-node] at (0, -.5) (y) {$2$};
        \node[lattice-node] at (1.5, -.5) (z) {$3$};
        \node[lattice-node] at (-1.5, .5) (xy) {$12$};
        \node[lattice-node] at (0, .5) (xz) {$13$};
        \node[lattice-node] at (1.5, .5) (yz) {$23$};
        \node[lattice-node] at (0, 1.5) (xyz) {$123$};

        \path[lattice-edge] (0) edge (x);
        \path[lattice-edge] (0) edge (y);
        \path[lattice-edge] (0) edge (z);
        \path[lattice-edge] (x) edge (xy);
        \path[lattice-edge] (x) edge (xz);
        \path[lattice-edge] (y) edge (xy);
        \path[lattice-edge] (y) edge (yz);
        \path[lattice-edge] (z) edge (xz);
        \path[lattice-edge] (z) edge (yz);
        \path[lattice-edge] (xy) edge (xyz);
        \path[lattice-edge] (xz) edge (xyz);
        \path[lattice-edge] (yz) edge (xyz);

        \hglabel{0}  {0}{+1}
        \hglabel{x}  {0}{-1}
        \hglabel{y}  {0}{-1}
        \hglabel{z}  {0}{-1}
        \hglabel{xy} {1}{+1}
        \hglabel{xz} {1}{0}
        \hglabel{yz} {1}{0}
        \hglabel{xyz}{1}{+1}
    \end{scope}

    \begin{scope}[shift={(8,-8)}]
        \sublattice{yellow!50!white}
        \node[gray!60!white, scale = 1.5] at (-.75, -.5) {$L_1$};
        \begin{scope}[shift={(1.5,1.05)}]
            \sublattice{cyan!30!white}
            \node[gray!60!white, scale = 1.5] at (-.75, -.5) {$L_2$};
        \end{scope}
        \node[rotate = -25] at (-1, -1.8) {\hgpair{h_1'}{g_1'}};
        \node[rotate = -25] at (1, 1.4) {\hgpair{h_2'}{g_2'}};

        \node[lattice-node] at (0, -1.5) (0) {$\emptyset$};
        \node[lattice-node] at (-1.5, -.5) (x) {$1$};
        \node[lattice-node] at (0, -.5) (y) {$2$};
        \node[lattice-node] at (1.5, -.5) (z) {$3$};
        \node[lattice-node] at (-1.5, .5) (xy) {$12$};
        \node[lattice-node] at (0, .5) (xz) {$13$};
        \node[lattice-node] at (1.5, .5) (yz) {$23$};
        \node[lattice-node] at (0, 1.5) (xyz) {$123$};

        \path[lattice-edge] (0) edge (x);
        \path[lattice-edge] (0) edge (y);
        \path[lattice-edge] (0) edge (z);
        \path[lattice-edge] (x) edge (xy);
        \path[lattice-edge] (x) edge (xz);
        \path[lattice-edge] (y) edge (xy);
        \path[lattice-edge] (y) edge (yz);
        \path[lattice-edge] (z) edge (xz);
        \path[lattice-edge] (z) edge (yz);
        \path[lattice-edge] (xy) edge (xyz);
        \path[lattice-edge] (xz) edge (xyz);
        \path[lattice-edge] (yz) edge (xyz);

        \hglabel{0}  {0}{0}
        \hglabel{x}  {0}{0}
        \hglabel{y}  {0}{0}
        \hglabel{z}  {0}{-1}
        \hglabel{xy} {0}{0}
        \hglabel{xz} {1}{0}
        \hglabel{yz} {1}{0}
        \hglabel{xyz}{1}{+1}
    \end{scope}

    \begin{scope}[shift={(2,-8)}]
        \node[lattice-node] at (0, -1.5) (0) {$\emptyset$};
        \node[lattice-node] at (-1.5, -.5) (x) {$1$};
        \node[lattice-node] at (0, -.5) (y) {$2$};
        \node[lattice-node] at (1.5, -.5) (z) {$3$};
        \node[lattice-node] at (-1.5, .5) (xy) {$12$};
        \node[lattice-node] at (0, .5) (xz) {$13$};
        \node[lattice-node] at (1.5, .5) (yz) {$23$};
        \node[lattice-node] at (0, 1.5) (xyz) {$123$};

        \path[lattice-edge] (0) edge (x);
        \path[lattice-edge] (0) edge (y);
        \path[lattice-edge] (0) edge (z);
        \path[lattice-edge] (x) edge (xy);
        \path[lattice-edge] (x) edge (xz);
        \path[lattice-edge] (y) edge (xy);
        \path[lattice-edge] (y) edge (yz);
        \path[lattice-edge] (z) edge (xz);
        \path[lattice-edge] (z) edge (yz);
        \path[lattice-edge] (xy) edge (xyz);
        \path[lattice-edge] (xz) edge (xyz);
        \path[lattice-edge] (yz) edge (xyz);

        \hglabel{0}  {0}{0}
        \hglabel{x}  {1}{0}
        \hglabel{y}  {1}{0}
        \hglabel{z}  {1}{-1}
        \hglabel{xy} {1}{-1}
        \hglabel{xz} {2}{0}
        \hglabel{yz} {2}{0}
        \hglabel{xyz}{2}{+2}

        \node[scale=1.5] at (0, -2.25) {\hgpair{h'}{g'}};
    \end{scope}

    \draw[->, >=stealth, line width = 4] (4.25, -2) -- (5.75, -2);
    \draw[->, >=stealth, line width = 4] (5.75, -8) -- (4.25, -8);
    \draw[->, >=stealth, line width = 4] (8,-4.25) -- (8,-5.75);

\end{tikzpicture}
\caption{Illustration of Example~\ref{ex:parity:g}.
    The top-left corner shows the lattice $L=2^{[3]}$, where each node is annotated with a pair $(h, g)$, which are the values of the original $h$ and $g$ of the parity function.
    The bottom-left corner shows the final $(h', g')$ satisfying the conditions of Lemmaa~\ref{lemma:h:domination}, including normality.}
\label{fig:example}
\end{figure*}

%% file: sec-discussion.tex
\section{Conclusion and Discussion}

\label{sec:discussion}

In this paper we established a fundamental connection between
information inequalities and query containment under bag semantics.
In particular, we proved that the max-information-inequality problem
is many-one equivalent to the query containment where the containing
query is acyclic.  It is open whether these problems are decidable.
Our results help in the sense that, progress on one of these open
questions will immediately carry over to the other.  We end with a
discussion of our results and a list of open problems.

{\bf Beyond Chordal} Our results showed that the query containment
problem $Q_1 \preceq Q_2$ is equivalent to a $\miip$ when $Q_2$ is
either acyclic, or when it is chordal {\em and} has a simple junction tree.  In all other cases, condition
(\ref{eqn:another:sufficient}) is only sufficient, and we do not know
if it is also necessary.
%
%
%

{\bf Repeated Variables, Unbounded Arities} Our reduction form $\miip$
to query containment constructs two queries $Q_1, Q_2$ where the atoms
have repeated variables, and the arities of some of the relation names
depend on the size of the $\miip$.  We leave open the question whether
the reduction can be strengthened to atoms without repeated variables,
and/or queries over vocabularies of bounded arity.

{\bf Max-Linear Information Inequalities} Linear information
inequalities have been studied extensively in the literature, while
Max-linear ones much less.  Our result proves the equivalence of
$\bcqca$ and $\miip$, and this raises the question of whether $\iip$
and $\miip$ are different.
\conferenceorfull{In the full version of the paper~\cite{full-version}, we provide some
  evidence suggesting that they might be the same.}
{The following theorem
(Appendix~\ref{appendix:lambdas}) suggests that they might be
computationally equivalent.

\begin{thm} \label{th:the:lambdas}
  Let $E_\ell$, $\ell = 1, m$ be linear expressions of entropic
  terms.  Then the following conditions are equivalent:
  \begin{itemize}
  \item This max-linear inequality holds:
    $\forall h \in \Gamma_n^*$, $0 \leq \max_\ell E_\ell(h)$.
  \item There exists $\lambda_\ell \geq 0$, s.t.
    $\sum_\ell \lambda_\ell = 1$ and, denoting
    $E \defeq \sum_\ell \lambda_\ell E_\ell$, this linear inequality
    holds: $\forall h \in  \Gamma_n^*$, $0 \leq E(h)$.
  \end{itemize}
\end{thm}

The second item implies the first, because
$\max_\ell E_\ell \geq \sum_\ell \lambda_\ell E_\ell$; the proof of
the other direction is in the appendix.  Suppose we could strengthen
the theorem and prove that the $\lambda$'s can be chosen to be
rationals.  Then there exists a simple Turing-reduction from the
$\miip$ to $\iip$: given a $\miip$, search in parallel for a counter
example (by iterating over all finite probability spaces), and for
rational $\lambda$'s such that
$\sum_\ell \lambda_\ell E_\ell(h) \geq 0$ (which can be checked using
the $\iip$ oracle).  However, we do not know if the $\lambda$'s can
always be chosen to be rational.
}

{\bf The remarkable formula $E_T$ (Eq.(\ref{eq:et}))} The first to
introduce the expression $E_T$ was Tony
Lee~\cite{DBLP:journals/tse/Lee87}.  This early paper established
several fundamental connections between the entropy $h$ of the uniform
distribution of a relation $P$, and constraints on $P$: it showed that
an FD $X \rightarrow Y$ holds iff $h(Y|X)=0$, that an MVD
$X \twoheadrightarrow Y$ holds iff $I(Y; V - (X \cup Y)|X) = 0$,
and, finally, that $P$ admits an acyclic join decomposition given by a
tree $T$ iff $E_T(h) = h( V)$.  It also proved that $E_T$ is
equivalent to an inclusion-exclusion expression, which, in our
notation becomes:
\begin{align}
  E_t = & \sum_{S \subseteq \nodes(T)} (-1)^{|S|+1} CC(T \cap S) \cdot h(\chi(S))
\label{eq:et:inclusion:exclusion}
\end{align}
where $\chi(S) \defeq \bigcap_{t \in S}\chi(t)$, and $CC(T \cap S)$
denotes the number of connected components of the subgraph of $T$
consisting of the nodes
$\setof{t}{t \in \nodes(T), \chi(t) \cap \bigcup_{t' \in S} \chi(t')
  \neq \emptyset}$.

{\bf Discussion of Kopparty and Rossman~\cite{HDE}} We now re-state
the results in~\cite{HDE} using the notions introduced in this paper
in order to describe their connection.  Theorem 3.1 in~\cite{HDE}
essentially states that Eq.(\ref{eqn:another:sufficient}) is
sufficient for containment, thus it is a special case of our
Theorem~\ref{th:sufficient} for graph queries; they use an
inclusion-exclusion formula for $E_T$, similar to
(\ref{eq:et:inclusion:exclusion}), but given for chordal queries only.
Theorem 3.2 in~\cite{HDE} essentially states that, if
Eq.(\ref{eqn:another:sufficient}) fails on a normal polymatroid, then
there exists a database $\calD$ witnessing $Q_1 \not\preceq Q_2$, thus
it is a special case of our Lemma~\ref{lemma:necessary:chordal} for
the case when the queries are graphs; they use a different expression
for $E_T$, based on the M\"obius inversion of $h$.  This inversion is
precisely the I-measure of $h$, as we explain in
Appendix~\ref{sec:background:it:long}.  Finally, Theorem 3.3 in~\cite{HDE} proves
essentially that Eq.(\ref{eqn:another:sufficient}) is necessary and
sufficient when $Q_1$ is series-parallel and $Q_2$ is chordal.  This
differs from our Theorem~\ref{th:decidable} in that it imposes more
restrictions on $Q_1$ and fewer on $Q_2$.  The proof of our
Theorem~\ref{th:decidable} relies on the fact that any counterexample
of Eq.(\ref{eqn:another:sufficient}) is a normal entropic function,
but this does not hold in the setting of Theorem 3.3~\cite{HDE};
however, the only exception is given by the parity function
(Appendix~\ref{sec:background:it:long}), a case that~\cite{HDE} handles directly.

%% file: appendix.tex
\appendix

\section{Background on CQ's}

\label{appendix:problem}

\conferenceorfull{We prove the following in the full version of the paper~\cite{full-version}:}{}

\begin{lmm} \label{lemma:simple:bag-semantics} The containment problem
  under bag-set semantics $Q_1 \preceq Q_2$ is reducible in polynomial
  time to the containment problem under bag-set semantics for Boolean
  queries, $Q_1' \preceq Q_2'$.  Moreover, this reduction preserves
  any property of queries discussed in this paper: acyclicity,
  chordality, simplicity.
\end{lmm}

\conferenceorfull{}
{
\begin{proof} Assume w.l.o.g. that $Q_1, Q_2$ have the same head
  variables $ x$ (rename them otherwise).  Define two Boolean
  queries $Q_1', Q_2'$ by adding new unary atoms $U_i(x_i)$ to
  $Q_1, Q_2$, one atom for each $x_i \in  x$.  We prove:
  $Q_1 \preceq Q_2 \Leftrightarrow Q_1' \preceq Q_2'$.  For the
  $\Rightarrow$ direction, fix a database instance $\calD'$, denote
  the product of the unary relations by $U \defeq \prod_i U_i^D$, and
  let $\calD$ be obtained from $\calD'$ by removing the unary
  relations $U^D_i$.  It follows that
  $\bigcup_{ d \in U} Q_\ell[ d](\calD) = \hom(Q_\ell',\calD')$, for
  $\ell=1,2$.  Since $Q_1 \preceq Q_2$, and the sets $Q_\ell[d](\calD), d \in U$
  are disjoint, for $\ell=1,2$, we conclude
  $|\hom(Q_1',\calD')| = \sum_{ d \in U} |Q_1[ d](\calD)| \leq
  \sum_{ d \in U} |Q_2[ d](\calD)| = |\hom(Q_2',\calD')|$.  For
  the $\Leftarrow$ direction, let $\calD$ be a database instance, and
  let $ d \in D^{ x}$.  Define $\calD'$ to be the database obtained
  by adding to $\calD$ unary relations with one element,
  $U_i^D \defeq \set{d_i}$ for each $x_i\in x$.  Then,
  $Q_\ell[ d](\calD)=\hom(Q_\ell',\calD')$ for $\ell=1,2$.  By assumption
  $Q_1' \preceq Q_2'$, which implies
  $|Q_1[ d](\calD)| = |\hom(Q_1',\calD')| \leq |\hom(Q_2',\calD')| =
  |Q_2[ d](\calD)|$.
\end{proof}

\begin{ex} \label{ex:chaudhuri:vardi} We illustrate with this example
  from~\cite{DBLP:conf/pods/ChaudhuriV93}:
  \begin{align*}
    Q_1(x,z) = & P(x) \wedge S(u,x), \wedge S(v,z) \wedge R(z)\\
    Q_2(x,z) = & P(x) \wedge S(u,y), \wedge S(v,y) \wedge R(z)
  \end{align*}
We associate them to the following two Boolean queries:
  \begin{align*}
    Q_1'() = & P(x) \wedge S(u,x), \wedge S(v,z) \wedge R(z) \wedge U_1(x)\wedge U_2(z)\\
    Q_2'() = & P(x) \wedge S(u,y), \wedge S(v,y) \wedge R(z) \wedge U_1(x)\wedge U_2(z)
  \end{align*}
  Then $Q_1 \preceq Q_2$ iff $Q_1' \preceq Q_2'$; the latter can be
  shown using Theorems~\ref{th:sufficient} and \ref{th:necessary}.
\end{ex}
}

We  prove now a claim that we made in
Sec~\ref{subsec:sufficient},
namely that, for any node $t$ of
a tree decomposition, we  can assume
$\vars(Q_t) = \chi(t)$, where $Q_t$ is the query obtained by taking
the conjunction of all atoms with $\vars(A)\subseteq \chi(t)$.

\begin{fact} (Informal) Let $(T,\chi)$ be a tree decomposition of some
  query $Q$, and, forall $t \in \nodes(T)$, let $Q_t$ denote the
  conjunction of $A\in\atoms(Q)$ s.t. $\vars(A) \subseteq \chi(t)$.  Then,
  for the purpose of query containment, we can assume that
  $\vars(Q_t) = \chi(t)$, for every $t \in \nodes(T)$.
  More specifically, we can assume that for every $t \in \nodes(T)$ and every $A\in\atoms(Q)$ such that $\vars(A)\cap \chi(t) \neq \emptyset$, there exists $A'\in\atoms(Q)$ such that
  $\vars(A') = \vars(A)\cap \chi(t)$, hence $A'\in\atoms(Q_t)$.
  \label{fct:BCQC:proj}
\end{fact}

\begin{proof} To see an example where this property fails, consider
  $Q = R(x,y,u)\wedge S(y,z)\wedge R(x,z,v)$.  Let $T$ be the tree
  decomposition $\set{x,y,u}-\set{x,y,z}-\set{x,z,v}$, and let $t$ be
  the middle node, $\chi(t) = \set{x,y,z}$.  Then $Q_t = S(y,z)$ and
  its variables do not cover $\chi(t)$.

  We prove that the property can be satisfied w.l.o.g.  We first
  modify the vocabulary, by adding for each relation name $R$ of arity $a$
  and for each $S \subset [a]$,
  a new relation name $R_S$ of arity $|S|$.  Similarly, we
  modify a query $Q$ by adding, for each atom $R(X_1, \ldots, X_a)$ and for each $S \subset [a]$,
  a new atom $R_S( x_S)$, where $ x_S \defeq (X_i)_{i\in S}$.  Denote by
  $\hat Q$ the modified query.  Obviously $\hat Q$ satisfies the
  desired property.  We claim that this change does
  not affect query containment, more precisely $Q_1 \preceq Q_2
  \Leftrightarrow \hat Q_1 \preceq \hat Q_2$.  The $\Leftarrow$
  direction follows by expanding an input database $\calD$ for $Q_1,
  Q_2$ with extra predicates $R_S^D \defeq \Pi_S(R^D)$ for every relation symbol $R$ and every $S\subset [a]$ where $a$ is the arity of $R$.  The
  $\Rightarrow$ direction follows from modifying an input database
  $\calD$ for $\hat Q_1, \hat Q_2$ by replacing every ($a$-ary) relation $R^D$
  by $R^D \ltimes \left(\Join_{S\subset [a]} R_S^D\right)$.
\end{proof}

\input{sec-background-it-full}

\conferenceorfull{}
{
\section{Proof of Theorem~\ref{th:the:lambdas}}

\label{appendix:lambdas}

We note that we can replace $\Gamma_n^*$ in the statement of
Theorem~\ref{th:the:lambdas} by $\bar \Gamma_n^*$, because any
max-information inequality holds on $\Gamma_n^*$ iff it holds on
$\bar \Gamma_n^*$.  We will then prove that the theorem holds more
generally, for any closed, convex cone $K$; the claim follows from the
fact that  $K = \bar \Gamma_n^*$ is a closed, convex cone.

Recall that a {\em cone} is a subset $K \subseteq \R^N$ such that
$x \in K$ implies $c\cdot x \in K$ forall $c \geq 0$.  All four sets
$\calM_n, \calN_n, \bar \Gamma_n^*, \Gamma_n$ defined in
Sec.~\ref{sec:background:it:short} are closed, convex cones.   We prove:

\begin{thm} \label{th:cone:max}
  Let $K \subseteq \R^N$ be a closed, convex cone, and let $y_1,
  \ldots, y_m \in R^N$.  Then the following two conditions are
  equivalent:
  \begin{enumerate}
      \item \label{item:th:cone:max:1} $\forall x \in K: \max_i \inner{x,y_i} \geq 0$.
  \item \label{item:th:cone:max:2} There exist
    $\lambda_1, \ldots, \lambda_m \geq 0$ such that
    $\sum_i \lambda_i = 1$ and
    $\forall x \in K: \sum_i \lambda_i \inner{x, y_i} \geq 0$.
    Equivalently, $\sum_i \lambda_i y_i \in K^*$ (the dual of $K$).
  \end{enumerate}
\end{thm}

Notice that the coefficients $\lambda_i$ need not necessarily be
rational numbers.  The theorem says that every $\miip$ can be
reduced to an $\iip$ with, possibly irrational coefficients.

\begin{proof}
  Obviously (\ref{item:th:cone:max:2}) implies
  (\ref{item:th:cone:max:1}) because $\max_i \inner{x, y_i} \geq \sum_i
  \lambda_i \inner{x,y_i}  \geq 0$.  We will prove that
  (\ref{item:th:cone:max:1}) implies   (\ref{item:th:cone:max:2}).

  First, we prove that   (\ref{item:th:cone:max:1}) implies
  (\ref{item:th:cone:max:2}) when $K$ is a finitely generated cone:
  $K=\setof{x}{Ax \geq 0}$ for some $P \times N$ matrix $A$.
  Condition (\ref{item:th:cone:max:1}) implies that the following
  optimization problem has a value $\geq 0$:

  \begin{align*}
      \texttt{minimize } & \max_i \ \inner{x, y_i} \\
    \texttt{where: } & Ax \geq 0 \\
                     & x \in \R^N
  \end{align*}

  This optimization problem is equivalent to the following, where
  $x_0$ is a fresh variable, and $B$ is the $m \times N$ matrix whose
  rows are the vectors $y_1, \ldots, y_m$:
  \begin{align*}
    \texttt{minimize } & x_0 \\
    \texttt{where: } & Ax \geq 0 && \texttt{$P$ rows} \\
                     &
\left[
                       \begin{array}{c}
                         x_0 \\
                         \ldots \\
                         x_0
                       \end{array}
\right] - Bx \geq 0 && \texttt{$m$ rows}
  \end{align*}
This is a linear optimization problem whose solution is equal to that
of the dual, which is a linear program over variables $\mu_1, \ldots,
\mu_P$, $\lambda_1, \ldots, \lambda_m$:
\begin{align*}
  \texttt{maximize } & 0 \\
  \texttt{where: } & \lambda_1 + \ldots + \lambda_m = 1 && \texttt{$x_0$}\\
                   & \mathbf{\mu}^t A - \mathbf{\lambda}^t B = 0 && \texttt{$x_1, \ldots, x_N$}\\
                   & \mathbf{\lambda} \geq 0, \mathbf{\mu} \geq 0
\end{align*}
Since the optimal solution of the primal is $\geq 0$, the dual must
have a feasible solution $\mathbf{\lambda}, \mathbf{\mu}$. To prove
Condition (\ref{item:th:cone:max:1}), assume $x \in K$.  Then
$A x \geq 0$, therefore $\mathbf{\mu}^t A x \geq 0$, thus
$\mathbf{\lambda}^t B x = \inner{\sum_i \lambda_i y_i, x} \geq 0$ proving
the theorem for the case when $K$ is a finitely generated cone.

We prove now the general case.  Let $K' = K \cap \Q$ be the vectors in
$K$ with rational coordinates, and let
$K' = \set{x_1, x_2, \ldots, x_n, \ldots}$ be an enumeration of $K'$.
For each $n \geq 0$, let $K_n \subseteq K$ be the closed, convex cone
generated by $\set{x_1, \ldots, x_n}$.  Let $\Lambda_n \subseteq \R^m$
be the set of all vectors $\mathbf{\lambda}$ satisfying Condition
(\ref{item:th:cone:max:1}) for the cone $K_n$.  Since $K_n$ is
finitely generated, we have $\Lambda_n \neq \emptyset$.  Furthermore
it is easy to check that $\Lambda_n$ is topologically closed.  Since
$\Lambda_n$ is bounded, it follows that $\Lambda_n$ is a compact
subset of $\R^N$.  Since
$K_1 \subseteq K_2 \subseteq \cdots \subseteq K_n \subseteq \cdots$ it
follows that
$\Lambda_1 \supseteq \cdots \supseteq \Lambda_n \supseteq \cdots$ This
implies that any finite family has a nonempty intersection:
$\Lambda_{n_1} \cap \cdots \cap \Lambda_{n_s} = \Lambda_{\max(n_1,
  \ldots, n_s)}\neq \emptyset$.  It follows that the entire family has
a non-empty intersection, i.e. there exists
$\mathbf{\lambda} \in \bigcap_{n \geq 0} \Lambda_n$.  We prove that
$\mathbf{\lambda}$ satisfies Condition (\ref{item:th:cone:max:1}).
Indeed, let $x \in K$, and consider any sequence $(x_n)_{n \geq 0}$
such that $x_n \in K_n$ and $\lim_n x_n = x$.  Forall $n \geq 0$,
$\mathbf{\lambda} \in \Lambda_n$, which implies
$\sum_i \lambda_i \inner{x_n, y_i} \geq 0$, therefore
$\sum_i \lambda_i <x, y_i> = \lim_n \sum_i \lambda_i <x_n, y_i> \geq
0$ proving the claim.
\end{proof}

}

%% file: sec-background-it-full.tex
\section{Background on Information Theory}
\label{sec:background:it:long}

In this section, we review some additional background in information theory used
in this paper, continuing the brief introduction in
Sec.~\ref{sec:background:it:short}.

\begin{fact} If $n=1$ (i.e. there is a single random variable) and $h$ is
  entropic, then $c \cdot h$ is also entropic for every $c > 0$.
\end{fact}

\begin{proof}
  Start with a distribution $p$ whose entropy is
  $\lceil c \rceil \cdot h$.  Let $n$ be the number of outcomes, and
  $p_1, \ldots, p_n$ their probabilities.  For each
  $\lambda \in [0,1]$ define $p^{(\lambda)}$ to be the distribution
  $p^{(\lambda)}_1 = p_1 + (1-p_1)(1-\lambda)$,
  $p^{(\lambda)}_i = p_i\cdot \lambda$ for $i > 1$, and $h^{(\lambda)}$
  its entropy.  Then $h^{(0)}=0$, $h^{(1)} = \lceil c \rceil \cdot h$,
  and, by continuity, there exists $\lambda$ s.t.
  $h^{(\lambda)} = c\cdot h$.
\end{proof}

\begin{cor}
  For every $W \subsetneq  V$ and every $c > 0$, the function
  $c \cdot h_W$ is entropic, where $h_W$ is the step function.  It
  follows that every normal function is entropic (because it is a sum
  $\sum_W c_W h(W)$ and $c_W h(W)$ is entropic).
\end{cor}

\begin{proof}
  By the previous fact, there exists a random variable $Z$ whose
  entropy is $h_0(Z) = c$.  Let $h$ be the entropy of the following
  $n$ random variables: forall $U \in { V} - W$, define
  $U \defeq Z$ (hence, forall $X \subseteq  V - W$,
  $h(X)=h_0(Z)=c$), and for every $U \in W$, define $U$ to be a
  constant (hence for every $X \subseteq W$, $h(X)=0$).  Therefore,
  $h = c \cdot h_W$.
\end{proof}

However, when $n \geq 3$, then Zhang and
Yeung~\cite{DBLP:journals/tit/ZhangY97} proved that $c \cdot h$ is not
necessarily entropic.  Their proof is based on the {\em parity
  function}, introduced in Example~\ref{ex:parity}.

\begin{fact}
    $\Gamma_3^*$ is not convex.
\end{fact}
\begin{proof}
Zhang and Yeung~\cite{DBLP:journals/tit/ZhangY97} prove this fact as follows.
Let $h$ be the entropy of the parity function in Example~\ref{ex:parity}.
For every $c > 0$, consider
the function $h' = c\cdot h$.  They prove that $h'$ is entropic iff
$c = \log M$, for some integer $M$, which implies that $\Gamma_3^*$ is
not convex.  We include here their proof for completeness.  Assuming
$h'$ is entropic let $p'$ be its probability distribution, then the
following independence constraints hold: $X \perp Y$, because
$h'(XY) = h'(X)+h'(Y)$, and similarly $X \perp Z$ and $Y \perp Z$.
The following functional dependencies also hold: $XY \rightarrow Z$
(because $h'(XY)=h'(XYZ)$) and similarly $XZ \rightarrow Y$,
$YZ \rightarrow X$.  Let $x,y,z$ be any three values s.t.
$p'(x,y,z) > 0$.  Then $p'(x,y,z) = p'(x,y) = p'(x)p'(y)$.  Similarly
$p'(x,y,z)=p'(y)p'(z)$, which implies $p'(x)=p'(z)$.  Therefore, for any
other value $x'$, $p'(x')=p'(z)$.  This means that the variable $X$ is
uniformly distributed, because $p'(x)=p'(x')$ forall $x, x'$, hence
$p'(x)=1/M$ where $M$ is the size of the domain of $X$.  It follows
that $h'(X) = \log M$, proving the claim.
\end{proof}

Yeung~\cite{Yeung:2008:ITN:1457455} proves that the topological
closure $\bar \Gamma_n^*$ is a convex set, for every $n$.  Thus,
$\Gamma_n^* \subseteq \bar \Gamma_n^*$ and the inclusion is strict for
$n \geq 3$.  The elements of $\bar \Gamma_n^*$ are called {\em almost
  entropic functions}.  We note that if a linear information
inequality, or a max-linear information inequality is valid forall
entropic functions $h \in \Gamma_n^*$, then, by continuity, it is also
valid forall almost entropic functions $h \in \bar \Gamma_n^*$.

Let $h$ be an entropic function, and $X, Y \subseteq  V$ two sets
of variables.  For every outcome $X=x$, we denote by $h(Y | X =x)$ the
entropy of $Y$ conditioned on $X=x$.  The function
$Y \mapsto h(Y | X=x)$ is an entropic function (by definition).
Recall that we have defined  $h(Y|X) \defeq h(XY) - h(X)$.  It can be
shown by direct calculation that $h(Y|X) = \sum_x h(Y|X=x)\cdot
p(X=x)$, in other words it is a convex combination of entropic
functions.  Thus, $h(Y|X)$ is the expectation, over the outcomes $x$,
of $h(Y|X=x)$, justifying the name ``conditional entropy''.

\begin{fact}
  In general, the mapping $Y \mapsto h(Y|X)$ is not entropic.
\end{fact}
\begin{proof}
  To see an example, consider two probability spaces on $X,Y,Z$, with
  probabilities $p,p'$ and entropies $h,h'$ such that $h$ is the
  entropy of the parity (Example~\ref{ex:parity}) and $h' = 2h$.
  Consider a 4'th variable $U$, whose outcomes are $U=0$ or $U=1$ with
  probabilities 1/2, and consider the mixture model: if $U=0$ then
  sample $X,Y,Z$ using $p$, if $U=1$ then sample $X,Y,Z$ using $p'$.
  Let $h''$ be the entropy over the variables $X,Y,Z,U$.  Then the
  conditional entropy $h''(W|U) = 3/2 h(W)$, for all
  $W \subseteq \set{X,Y,Z}$, and thus it is not entropic.
\end{proof}

\eat{
We now explain the connection to the I-measure.
\begin{cor} \label{cor:not:n3}
  The parity function is not in $\calN_3$.  It follows that
  $\calN_n \subsetneq \bar \Gamma_n^*$, for $n \geq 3$.
\end{cor}
To see this, it suffices to compute the inverse of the parity function
which we show here:
\begin{align*}
  &
    \begin{array}{|c|c|ccc|ccc|c|} \hline
      W= & \emptyset &X & Y & Z & XY & XZ & YZ & XYZ\\ \hline
    h(W)= & 0 & 1 & 1 & 1 & 2 & 2 & 2 & 2 \\
    g(W)= & 1 & -1 & -1 & -1 & 0 & 0 & 0 & 2 \\ \hline
    \end{array}
\end{align*}
}

Yeung~\cite{Yeung:2008:ITN:1457455} defines the I-measure as follows.
Fix a set of variables $V$, which we identify with $[n]$.  Let
$\Omega = 2^{[n]}-\set{\emptyset}$.  An I-measure is any function
$\mu : 2^\Omega \rightarrow \R$ such that
$\mu(X \cup Y) = \mu(X) + \mu(Y)$ whenever $X \cap Y = \emptyset$.
Notice that $\mu$ is not necessarily positive.  For each variable
$V_i \in V$ we denote by
$\hat V_i \defeq \setof{\omega \in \Omega}{i \in \omega} \subseteq
\Omega$, and extend this notation to sets $X \subseteq V$ by setting
$\hat X \defeq \bigcup_{V \in X} \hat V$.  For each variable $V_i$
denote $\hat V_i^1 \defeq \hat V_i$ and $\hat V_i^0\defeq$ the
complement of $\hat V_i$.  An {\em atomic cell} is an intersection
$C \defeq \bigcap_{j=1,n} \hat V_j^{\varepsilon_j}$, where
$\varepsilon_j \in \set{0,1}$ forall $j$, where at least one
$\varepsilon_j=1$.  Obviously, $\mu$ is uniquely defined by its values
on the atomic cells.

Given $h \in \R^{2^n}$ (not necessarily entropic), the {\em I-measure
  associated} to $h$ is the unique $\mu$ satisfying the following,
forall $X \subseteq V$:
\begin{align}
  h(X) = & \sum_{C: C \subseteq \hat X} \mu(C) \label{eq:mu}
\end{align}

The normal entropic functions $\calN_n$ are precisely those with a
non-negative I-measure.  This can be seen immediately by observing
that, for any step function $h_W$, it's I-measure $\mu_W$ assigns the
value 1 to the cell
$(\bigcap_{V \not\in W} V^1) \cap (\bigcap_{V \in W} V^0)$, and 0 to
everything else.  In fact, there is a tight connection between the
I-measure $\mu$ and the M\"obius inverse function $g$
(Eq.\eqref{eq:hg} in Sec.~\ref{sec:domination}), which we explain
next.  First, we notice that  Equation~\eqref{eq:hg} implies:
%
\begin{align}
  h(X) = & - \sum_{Y: Y \not\supseteq X} g(Y) \label{eq:g:alternative}
\end{align}
%
The connection between $\mu$ and $g$ follows by a careful inspection
of Eq.~\eqref{eq:mu} and Eq~\eqref{eq:g:alternative}.  Each atomic
cell $C$ in Eq.~\eqref{eq:mu} is uniquely defined by the set of its
negatively occurring variables, denote this by $\texttt{neg}(C)$.
Then, $C \subseteq \hat X$ iff $X \not\subseteq
\texttt{neg}(C)$. Define the function $g : 2^{ V} \rightarrow \R$ as
$g(\texttt{neg}(C)) \defeq -\mu(C)$ and $g( V) = h( V)$ (recall that
$\texttt{neg}(C) \neq V$).  Then Eq.(\ref{eq:mu}) becomes
$h(X) = \sum_{C: X \not\subseteq \texttt{neg}(C)} \mu(C) = - \sum_{Y:
  X \not\subseteq Y}g(Y)$ which is precisely
Eq.(\ref{eq:g:alternative}).

We end our background with a proof that the $\miip$ problem is
co-recursively enumerable.  Recall that a set $A \subseteq \Z^k$ is
called {\em recursively enumerable}, or r.e., if there exists a
Turning computable function $f$ whose image is $A$.  Equivalently,
there exists a computable function that, given $ x \in \Z^k$
returns ``true'' if $ x \in A$ and does not terminate if
$ x \not\in A$.  The set $A$ is called {\em co-recursively
  enumerable}, or co-r.e., if its complement is r.e.

\begin{lmm}
  $\miip$ is co-r.e.
\end{lmm}

\begin{proof} (Sketch) Enumerate all finite probability distributions
  where the probabilities are given by rational numbers, and check
  Eq.(\ref{eq:miti:0}) on each of them.  This is possible because each
  entropy value $h(X)$ is the log of a number of the form $\prod_i \left(\frac{1}{p^{(X)}_i}\right)^{p^{(X)}_i}$, where $i$ ranges over all possible assignments of the variable set $X$, and $p_i^{(X)}$ is the probability that $X$ takes the $i$-th assignment.
  Therefore the inequality becomes
  \[\exists \ell  \in [k] \quad\text{s.t.}\quad
  \prod_{X\subseteq V}\prod_i \left(\frac{1}{p^{(X)}_i}\right)^{c_{\ell, X}\cdot p^{(X)}_i} \geq 1.\]
  In the above, $c_{\ell,X}$ are integers while $p^{(X)}_i$ are rational numbers.
  We can raise both sides of the above inequality to a power of $d$, which is the common denominator among all $p^{(X)}_i$.
  If the inequality fails, then return ``false'', otherwise continue with the next
  finite probability distribution.
\end{proof}